\documentclass{article}
\usepackage{fullpage}
\usepackage{amssymb, amsmath, amsthm}
\usepackage{graphicx}
\usepackage{xspace}
\usepackage{lineno}
\usepackage{todonotes}
\usepackage[all=normal,bibliography=tight]{savetrees}
\usepackage{enumerate}
\usepackage[hidelinks]{hyperref}
\usepackage[capitalize]{cleveref}
\usepackage[utf8]{inputenc}
\usepackage{thmtools}
\usepackage{thm-restate}

\newtheorem{theorem}{Theorem}[section]
\newtheorem{lemma}[theorem]{Lemma}

\newtheorem{corollary}[theorem]{Corollary}
\theoremstyle{definition}
\newtheorem{observation}[theorem]{Observation}
\newtheorem{definition}[theorem]{Definition}

\newcommand{\Oh}{\mathcal{O}}
\newcommand{\Emb}{\mathcal{E}}
\newcommand{\tEmb}{{\widehat{\mathcal{E}}}}
\newcommand{\Eperm}{\theta}
\newcommand{\Vperm}{\sigma}
\newcommand{\Fperm}{\phi}
\newcommand{\Flags}{\mathbf{F}}
\newcommand{\orb}{\mathrm{orb}}
\newcommand{\genus}{\hat{g}}
\newcommand{\conncomp}{\mathrm{cc}}

\newcommand{\subgroup}[1]{\langle #1 \rangle}
\newcommand{\SymGrp}[1]{\mathrm{Sym}(#1)}
\newcommand{\Vertices}{\mathtt{Verts}}
\newcommand{\Edges}{\mathtt{Edges}}
\newcommand{\Faces}{\mathtt{Faces}}
\newcommand{\Label}{L}
\newcommand{\dom}{\mathrm{domain}}
\newcommand{\tdbag}{\beta}
\newcommand{\tdunder}{\alpha}
\newcommand{\tdG}{G^\downarrow}
\newcommand{\tdemb}{\mathrm{bnd}}
\newcommand{\DP}{\mathbf{T}}
\newcommand{\genusVD}{\textsc{Genus Vertex Deletion}}
\newcommand{\planarVD}{\textsc{Planar Vertex Deletion}}

\title{Deleting vertices to graphs of bounded genus}

\author{Tomasz Kociumaka\thanks{University of Warsaw, Poland, \texttt{kociumaka@mimuw.edu.pl}.}
   \and Marcin Pilipczuk\thanks{University of Warsaw, Poland, \texttt{malcin@mimuw.edu.pl}. Research supported by Polish National Science Centre grant DEC-2012/05/D/ST6/03214.}}
\date{}

\begin{document}
\maketitle
\begin{abstract}
We show that a problem of deleting a minimum number of vertices from a graph to obtain a graph embeddable on a surface of a given Euler genus
is solvable in time $2^{C_g \cdot k^2 \log k} n^{\Oh(1)}$, where $k$ is the size of the deletion set, $C_g$ is a constant depending on the Euler genus $g$ of the target surface,
and $n$ is the size of the input graph.
On the way to this result, we develop an algorithm solving the problem in question in time $2^{\Oh((t+g) \log (t+g))} n$, given a tree decomposition of the input graph of width $t$.
The results generalize previous algorithms for the surface being a sphere by Marx and Schlotter~\cite{ildi}, Kawarabayashi~\cite{kenichi}, and Jansen, Lokshtanov, and Saurabh~\cite{jls}.
\end{abstract}

\section{Introduction}

In recent years, a significant effort has been put into
the study of parameterized complexity of recognizing near-planar graphs~\cite{ildi,kenichi,jls},
that is, graphs that become planar after deleting a small number of vertices.
Since, by the classic result of Lewis and Yannakakis~\cite{lewis}, the decision version of the
problem is NP-hard, it is natural to look for fixed-parameter algorithms with various parameters.

The parameter of the size of the deletion set naturally comes from the supposed applications:
a number of efficient algorithms for planar graphs generalize well to near-planar graphs,
if one supply them with the deletion set. 
Formally, we define the problem \planarVD{} as follows: given a graph $G$ and an integer $k$,
decide, if one can delete at most $k$ vertices from $G$ to obtain a planar graph.

Clearly, for a fixed integer $k$, the yes-instances to \planarVD{} form a minor-closed 
graph class. Consequently, from the Graph Minors theory we obtain a nonuniform
fixed-parameter algorithm for \planarVD{} (cf.~\cite[Section~6]{ksiazka}).

Marx and Schlotter~\cite{ildi} showed an explicit, uniform fixed-parameter algorithm
by a typical irrelevant vertex approach. First, they observe that the
formulation of the problem as hitting all models of the forbidden minors for planar graphs (i.e., $K_5$ and $K_{3,3}$) leads to a fixed-parameter algorithm on bounded treewidth graphs by relatively standard techniques. Second, it is quite easy to believe (but quite technical to formally prove)
that a middle part of a large, flat (planar), and grid-like subgraph of the input graph will never
be part of an optimal deletion set, and can be removed without changing the answer to the problem.
The combination of the excluded grid theorem and the technique of iterative compression gives
here a win-win approach: if the treewidth of the graph is not sufficiently bounded, an irrelevant
part can be uncovered and removed.

There are two sources of potential inefficiencies in the approach of Marx and Schlotter.
First, the routine for graphs of bounded treewidth that finds a minimum set hitting all forbidden minor models work in time double-exponential in the treewidth bound. 
Since the treewidth bound needs to be significantly larger than the size of the deletion set
for the irrelevant vertex argument to work, we obtain at least a double-exponential
dependency on the parameter.
Second, the technique of iterative compression, at least applied in a straightforward manner,
  gives at least quadratic dependency on the input size.

Later, Kawarabayashi~\cite{kenichi} showed a fixed-parameter algorithm with linear dependency 
on the input size. Finally, Jansen, Lokshtanov, and Saurabh~\cite{jls} showed an algorithm with 
running time $2^{\Oh(k \log k)} n$, that is, with nearly single-exponential dependency
on the parameter and linear dependency on the input size.

On high level, the work of~\cite{jls} follows the approach of Marx and Schlotter, but
improves upon both components. First, they show that a routine that explicitly constructs
partial embeddings of graphs of bounded treewidth solves the problem in question
in time $2^{\Oh(t \log t)} n$, given a tree decomposition of width $t$. 
Second, they show arguments in the spirit of the aforementioned irrelevant vertex rule
that reduce the graph to treewidth linearly bounded in the size of the solution (deletion set).
Third, they apply a more involved iterative compression approach that in one step compresses
the graph by a multiplicative factor, yielding a linear dependency on the input size.

A simple reduction from the \textsc{Vertex Cover} problem (replace every edge $uv$ with a
$K_5$ with vertices $u$, $v$, and $3$ new vertices) shows that, unless the
Exponential Time Hypothesis (ETH)~\cite{DBLP:journals/jcss/ImpagliazzoPZ01} fails, the dependency on the parameter $k$
needs to be $2^{\Omega(k)}$ for any parameterized algorithm for \planarVD{}.
Although it is open whether \planarVD{} can be solved in $2^{o(k \log k)} n^{\Oh(1)}$ time,
we note that 
a lower bound by the second author~\cite{tw-lb} asserts that, unless the ETH fails,
the bounded treewidth subroutine requires dependency $2^{\Omega(t \log t)}$ on the treewidth
of the graph. This implies that a hypothetical algorithm that solves \planarVD{}
in  $2^{o(k \log k)} n^{\Oh(1)}$ time needs to follow significantly different approach
than the one we know currently.

In the light of the aforementioned developments, in this paper, we initiate the study
of the \genusVD{}: given a graph $G$ and integers $g$ and $k$, decide, if one can delete
at most $k$ vertices from $G$ to obtain a graph embeddable on a surface of Euler genus at most $g$.

Our main result is the following:
\begin{theorem}\label{thm:main}
\genusVD{} can be solved in time $2^{C_g k^2 \log k} n^{\Oh(1)}$, where
$C_g$ is a constant depending on $g$ only.
\end{theorem}
    
In the proof of Theorem~\ref{thm:main}, we follow the general approach of Marx and Schlotter~\cite{ildi}.
Our main technical contribution is the generalization of the bounded treewidth subroutine of~\cite{jls} to the bounded genus case.

\begin{theorem}\label{thm:genusVDtw}
Given a \genusVD{} instance $(G,g,k)$ with $|V(G)| = n$, and a tree decomposition of $G$ of width $t$, one can solve the \genusVD{} problem on $(G,g,k)$ in time $2^{\Oh((t+g) \log (t+g))} n$.
\end{theorem}

The proof of Theorem~\ref{thm:genusVDtw} follows the same principle as the corresponding
routine of~\cite{jls} --- building partial embeddings for graphs with small boundaries ---
but requires significant technical hurdle to be presented formally and in full detail.

The task of the algorithm of Theorem~\ref{thm:genusVDtw} can be also seen as a generalization
of an algorithm that computes the Euler genus in graphs of bounded treewidth: by substituting
$k=0$ we obtain an algorithm with running time $2^{\Oh((t+g) \log (t+g))} n$ that checks if the input
graph is embeddable on a surface of Euler genus at most $g$.
Such a result is not new:
a bounded treewidth routine is also part of the current algorithms that compute embeddings
in linear time~\cite{mohar1,mohar2}. 
In particular, the work of Kawarabayashi, Mohar, and Reed~\cite{mohar2} 
claims an algorithm with running time $f(t) \cdot n$ for some function $f$
(i.e., without the exponential dependency on $g$).
However, the work~\cite{mohar2} is an extended abstract from FOCS 2008 that, to the best
of our knowledge, never substantiated in a full version, and we were unable to reproduce
the details of this algorithm from the description in~\cite{mohar2}.

For the second part, namely the irrelevant vertex argument,
we generalize the arguments of Marx and Schlotter~\cite{ildi}:

\begin{restatable}{theorem}{irr}\label{thm:irr}
For every integer $g$ there exists a constant $C_g$ such that the following holds.
Given a graph $G$, an integer $k$, and a set $M \subseteq V(G)$ such that $G-M$ is embeddable into a surface of genus at most $g$,
one can in time $C_g n^{\Oh(1)}$ find one of the following:
\begin{enumerate}
\item a tree decomposition of $G$ of width at most $C_g |M|^{1/2} k^{3/2}$;
\item a vertex $w \in V(G)$ such that every solution to the \genusVD{} instance $(G,g,k)$ contains $w$; or
\item a vertex $v \in V(G)$ such that $(G, g, k)$ is a yes-instance to \genusVD{} if and only if $(G-\{v\}, g, k)$ is.
\end{enumerate}
\end{restatable}

Theorem~\ref{thm:main} follows now from Theorems~\ref{thm:genusVDtw} and~\ref{thm:irr}
in a standard manner.
By a standard application of the irrelevant vertex technique (cf.~\cite[Section 4]{ksiazka},
we can assume that, apart from the input \genusVD{} instance $(G,g,k)$,
we are additionally given a set $M \subseteq V(G)$ of size $k+1$ such that $G-M$ is embeddable
on a surface of Euler genus at most $g$ (i.e., $M$ is a solution of slightly too large size),
at the cost of an additional $\Oh(n)$ factor in the running time bound.
We iteratively apply the algorithm of Theorem~\ref{thm:irr} to $(G,g,k)$ and $M$:
if a vertex $w$ or $v$ is returned, we delete it and restart (decreasing the parameter $k$ by one in case of a vertex $w$);
if a tree decomposition is returned, we solve the \genusVD{} problem with the algorithm of Theorem~\ref{thm:genusVDtw}.

Since $|M|=k+1$, the algorithm of Theorem~\ref{thm:genusVDtw} is applied to a tree decomposition
of width of the order of $C_g k^2$, yielding the bound promised in Theorem~\ref{thm:main}.

While the lower bound of~\cite{tw-lb} shows that the bounded-treewidth routine of Theorem~\ref{thm:genusVDtw}
has optimal running time (assuming ETH), we conjecture that 
the running time bound of Theorem~\ref{thm:main} is not optimal, and can be improved similarly as it was in the planar case~\cite{jls}.
For this reason, we do not optimize the parameter dependency in Theorem~\ref{thm:irr}, favouring the clarity of the arguments.
In other words, we view our contribution as Theorem~\ref{thm:genusVDtw} being the main technical merit, while
Theorem~\ref{thm:irr} and the resulting Theorem~\ref{thm:main} being an example application.

\medskip

The paper is organized as follows: after introducing notation for permutations and
tree decompositions in Section~\ref{sec:prelims}, we discuss combinatorial embeddings
in Section~\ref{sec:embed}. Theorem~\ref{thm:genusVDtw} is proved in Section~\ref{sec:tw},
and Theorem~\ref{thm:irr} is proved in Section~\ref{sec:irr}.

\section{Preliminaries}\label{sec:prelims}

\subsection{Permutations, involutions, cycles}

For a nonnegative integer $t$, we denote $[t] = \{1,2,\ldots,t\}$.
The group of all permutations of a set $U$ is denoted by $\SymGrp{U}$.
Given a set of permutations $S$ of a set $U$, by $\subgroup{S}$ we denote the subgroup of $\SymGrp{U}$ generated by $S$.
Given a subgroup $\Gamma$ of the group of $\SymGrp{U}$, by $\orb(\Gamma)$ we denote the family of orbits of $\Gamma$.
For a permutation $\sigma$, $\orb(\sigma)$ is a shorthand for $\orb(\subgroup{\sigma})$; this also allows us to speak about orbits of a single permutation.
An orbit is \emph{trivial} if it consists of a single element.

A permutation $\sigma$ is an \emph{involution} if $\sigma(\sigma(i)) = i$ for every $i \in \dom(\sigma)$
and is \emph{fixed-point free} if $\sigma(i) \neq i$ for every $i \in \dom(\sigma)$.
Note that a permutation is a fixed-point free involution if and only if every its orbit is of size exactly two.

A permutation $\sigma$ is a \emph{cycle permutation} if it has exactly one nontrivial orbit; note that $\sigma$ needs to act cyclically on this orbit.
A \emph{cycle} is an unordered pair consisting of a cycle permutation and its inverse.

We will need the operation of \emph{restricting} a cycle permutation $\sigma$ to a subset $A$ of the elements of the nontrival orbit $v$ of $\sigma$:
the result is a permutation $\sigma_A$, where $\sigma_A(e) = e$ for every $e \notin A$ while for every $e \in A$ we have $\sigma_A(e) = \sigma^k(e)$ where $k$ is the minimum positive
integer with $\sigma^k(e) \in A$. In other words, we shorten the nontrivial orbit $v$ by crossing out the elements not belonging to $A$. 
Note that if $|A| \geq 2$, then $\sigma_A$ is also a cycle permutation, while for $|A| \leq 1$ we have $\sigma_A$ being an identity.
The definition of restricting naturally extends to cycles by restricting both components.

For a sequence $P$ of elements of some set, by $\bar{P}$ we denote its reverse. A \emph{subsequence} of a cycle is a sequence consisting of consecutive elements of the nontrivial orbit of the cycle.

In our work we will often analyze the subgroup of the permutation group spanned by two fixed-point free involutions. 
Let $\alpha$ and $\beta$ be two fixed-point free involutions of a set $U$ and let $\Gamma = \subgroup{\alpha, \beta}$.
Observe that for any orbit $v$ of $\Gamma$ and any element $e \in v$, the orbit $v$ consists of elements
$$e, \alpha(e), \beta(\alpha(e)), \alpha(\beta(\alpha(e))), \beta(\alpha(\beta(\alpha(e)))), \alpha(\beta(\alpha(\beta(\alpha(e))))),\ldots.$$
The above order is cyclic: $|v|$ is always even, and if $|v| = 2k$, then $(\beta \circ \alpha)^k (e) = e$.
To fix notation, let us define a cycle permutation $o_e^{\subgroup{\alpha,\beta}}$ of $U$ as follows:
\begin{equation*}
o_e^{\subgroup{\alpha,\beta}}(f) = \begin{cases} 
  \alpha(f) & \mathrm{if}\ f \in v\ \mathrm{and}\ f = (\beta \circ \alpha)^i(e)\ \mathrm{for\ some}\ $i$,\\
  \beta(f) & \mathrm{if}\ f \in v\ \mathrm{and}\ f = \alpha \circ (\beta \circ \alpha)^i(e)\ \mathrm{for\ some}\ $i$,\\
  f & \mathrm{if} f \notin v.
\end{cases}
\end{equation*}
Note that $o_e^{\subgroup{\alpha,\beta}}$ is a cycle permutation whose nontrivial orbit is $v$. Furthermore, while its definition formally depends on the choice of $e$
and the order of $\alpha$ and $\beta$, different choices of $e \in v$ and a potential swap of the roles of $\alpha$ and $\beta$ lead either to $o_e^{\subgroup{\alpha,\beta}}$ or its inverse.
The definition of a cycle is suited to accommodate that: For the cycle permutation $o_e^{\subgroup{\alpha,\beta}}$, its cycle is denoted as $\hat{o}_e^{\subgroup{\alpha,\beta}}$.
By the previous argumentation, the cycle does not depend on the order of $\alpha$ and $\beta$, nor of the choice of the element $e$ within the same orbit.

Let $\hat{\sigma}_1$ and $\hat{\sigma}_2$ be two cycles of disjoint sets $U_1$ and $U_2$, with nontrivial orbits $v_1$ and $v_2$.
A \emph{merge} of $\hat{\sigma}_1$ and $\hat{\sigma}_2$ is a cycle $\hat{\sigma}$ of $U := U_1 \cup U_2$, whose nontrivial orbit $v$ consists of exactly elements of $v_1 \cup v_2$.
Furthermore, for some choice of cycle permutations $\sigma_1 \in \hat{\sigma}_1$, $\sigma_2 \in \hat{\sigma}_2$, $\sigma \in \hat{\sigma}$,
for every $i=1,2$ and $e \in v_i$, if $k$ is the minimum positive integer for which $\sigma^k(e) \in v_i$, then $\sigma^k(e) = \sigma_i(e)$.
In other words, if we restrict the cycle of the nontrivial orbit of $\hat{\sigma}$ to the elements of $U_i$ only, we obtain the cyclic order of the nontrivial orbit of $\hat{\sigma}_i$.

\subsection{Tree decompositions}

Given a graph $G$, a \emph{tree decomposition} of $G$ is a pair $(T, \tdbag)$ where $T$ is a rooted tree and is a function $\tdbag : V(T) \to 2^{V(G)}$
that assigns to every $t \in V(T)$ a \emph{bag} $\tdbag(t) \subseteq V(G)$ such that the following holds:
\begin{itemize}
\item for every edge $e \in E(G)$, there is a node $t \in V(T)$ with $e \subseteq \tdbag(t)$;
\item for every vertex $v \in V(G)$, the set $\{t \in V(T) : v \in \tdbag(t)\}$ is nonempty and induces a connected subgraph of $T$.
\end{itemize}
The width of a decomposition is the maximum size of a bag, minus one.

For a tree decomposition $(T,\tdbag)$ of a graph $G$, we define two auxiliary functions $\tdunder$ and $\tdG$. For a node $t \in V(T)$, 
\begin{itemize}
\item by $\tdunder(t) \subseteq V(G)$ we denote the union of all bags of descendants of $t$ (including $t$ itself);
\item by $\tdG(t)$ we denote the graph $G[\tdunder(t)] -E(G[\tdbag(t)])$, that is, the graph induced by $\tdG(t)$ with the edges
inside the bag $\tdbag(t)$ removed.
\end{itemize}

For dynamic programming algorithms, it is often convenient to work with so-called \emph{nice} tree decompositions, where the bag of the root is empty,
and every node $t \in V(T)$ is of one of the following four types:
\begin{description}
    \item[leaf node] has no children and its bag is empty;
    \item[introduce node] has one child $t'$ such that $\tdbag(t) = \tdbag(t') \cup \{v\}$ for some vertex $v \notin \tdbag(t')$;
    \item[forget node] has one child $t'$ such that $\tdbag(t) = \tdbag(t') \setminus \{v\}$ for some vertex $v \in \tdbag(t')$;
    \item[join node] has two children $t_1$ and $t_2$ with $\tdbag(t) = \tdbag(t_1) = \tdbag(t_2)$.
\end{description}
It is well-known (see, e.g.,~\cite{ksiazka,Kloks94})
that, in polynomial time, one can turn any tree decomposition into an equivalent nice one without increasing its width.

\section{Embeddings and operations on them}\label{sec:embed}

\subsection{Graph and hypergraph embeddings}

We start with a clean but abstract notion of a hypergraph embedding, and then we restrict ourselves only to graph embeddings.

\begin{definition}[hypergraph embedding]
A \emph{hypergraph embedding} is a tuple
$(\Flags, \Eperm, \Vperm, \Fperm)$, where $\Flags$ is a finite set, whose elements are called \emph{flags},
  and $\Eperm$, $\Vperm$, and $\Fperm$ are three fixed-point free involutions of the set $\Flags$.
\end{definition}

Given a hypergraph embedding $\Emb = (\Flags, \Eperm, \Vperm, \Fperm)$, we use the notation $\Flags(\Emb) := \Flags$ etc.
Note that in any hypergraph embedding, since $\Eperm, \Vperm$, and $\Fperm$ are fixed-point free involutions, the number of flags needs to be even.

\begin{definition}[connected components, vertices, edges, faces]
Let $\Emb = (\Flags, \Eperm, \Vperm, \Fperm)$ be a hypergraph embedding. Then
\begin{description}
\item[a connected component] is an orbit of $\conncomp_\Emb := \subgroup{\Eperm, \Vperm, \Fperm}$;
\item[a vertex] is an orbit of $\Vertices_\Emb := \subgroup{\Vperm, \Fperm}$;
\item[an edge] is an orbit of $\Edges_\Emb := \subgroup{\Eperm, \Vperm}$;
\item[a face] is an orbit of $\Faces_\Emb := \subgroup{\Eperm, \Fperm}$.
\end{description}
\end{definition}

Note that, as all three permutations of $\Emb$ are fixed-point free involutions, every vertex, edge, or a face of $\Emb$
can be identified with a cycle, whose nontrivial orbit is the vertex/edge/face in question.

Given a hypergraph embedding $\Emb = (\Flags, \Eperm, \Vperm, \Fperm)$, a \emph{size-$k$} edge, vertex, or face is an orbit of $\Edges_\Emb$, $\Vertices_\Emb$, or $\Faces_\Emb$ that consists of $k$ flags.
Note that $k$ is always a positive even integer in this context.
Also, observe that a size-$2$ edge corresponds to two equal orbits or $\Eperm$ and $\Vperm$, a size-$2$ vertex corresponds to two equal orbits of $\Vperm$ and $\Fperm$,
while a size-$2$ face corresponds to two equal orbits of $\Eperm$ and $\Fperm$.

We now define the genus of an embedding.
\begin{definition}[genus of a hypergraph embedding]
Given a hypergraph embedding $\Emb = (\Flags, \Eperm, \Vperm, \Fperm)$,
      its \emph{genus} is defined as
\begin{equation}\label{eq:genus-hyp}
\genus_\Emb := \tfrac{1}{2} |\Flags| - |\orb(\Vertices_\Emb)| - |\orb(\Edges_\Emb)| - |\orb(\Faces_\Emb)| + 2|\orb(\conncomp_\Emb)|.
\end{equation}
\end{definition}
Since in any hypergraph embedding the number of flags is even, the genus of a hypergraph embedding is always an integer. 

In a graph embedding, every edge consists of four flags.
\begin{definition}[graph embedding]
A hypergraph embedding $\Emb$ is a \emph{graph embedding}, or simply an \emph{embedding}, if every edge consists of exactly four flags. 
\end{definition}
In other words, a hypergraph embedding $\Emb = (\Flags, \Eperm, \Vperm, \Fperm)$ is a graph embedding if no two orbits of $\Eperm$ and $\Vperm$ coincide,
but the involutions $\Eperm$ and $\Vperm$ commute.

Observe that the formula for genus simplifies in case of a graph embedding.
\begin{observation}
If $\Emb = (\Flags, \Eperm, \Vperm, \Fperm)$ is an embedding, then its genus equals
\begin{equation}\label{eq:genus}
\genus_\Emb :=  |\orb(\Edges_\Emb)| - |\orb(\Vertices_\Emb)| - |\orb(\Faces_\Emb)| + 2|\orb(\conncomp_\Emb)|.
\end{equation}
\end{observation}
From this point, we use only graph embeddings in this work, and call them simply \emph{embeddings}.

If two objects (face, edge, vertex) share a flag, we say that these objects are \emph{incident}.
Note that an object with $k$ flags can be incident to at most $k/2$ objects of each of the other types (e.g., a size-$k$ vertex can be incident to at most $k/2$ edges and $k/2$ faces).
In particular, in an embedding, edge can be incident to one or two faces and one or two vertices. An edge incident with only one vertex is a \emph{loop}.

Let us now relate the aforementioned definition of a (combinatorial) embedding with the natural intuition. Let $G$ be a graph, embedded on a surface.
We visualize every edge as a (thin, and possibly bend) rectangle, with a flag attached at every corner of the rectangle (see Figure~\ref{fig:flags}).
The involution $\Eperm$ pairs up flags on an edge that lie on the same side.
The involution $\Vperm$ pairs up flags on an edge that lie at the same endpoint.
Finally, the involution $\Fperm$ pairs up neighboring flags of consecutive edges around a vertex.
In this manner, an orbit of $\Vertices_\Emb = \subgroup{\Vperm, \Fperm}$ yields a cyclic order of flags around a vertex,
an orbit of $\Edges_\Emb = \subgroup{\Eperm,\Vperm}$ yields a cyclic order of the four flags of an edge,
while an orbit of $\Faces_\Emb = \subgroup{\Eperm, \Fperm}$ yields a cyclic order of flags around a face.
The formula~\eqref{eq:genus} corresponds to the standard notion of an Euler genus of an embedding.

\begin{figure}
\begin{center}
\includegraphics{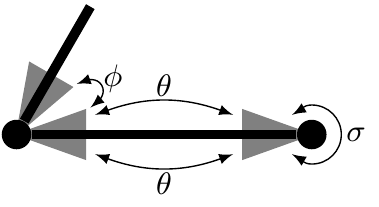}
\caption{Flags (gray triangles) and involutions in an embedding.}
\label{fig:flags}
\end{center}
\end{figure}

\subsection{Basic operations on embeddings}

We will need an operation of \emph{deleting an edge}, and a reverse operation of \emph{drawing a new edge}.

\subsubsection{Deleting an edge}

We start with the former. Let $\Emb = (\Flags, \Eperm, \Vperm, \Fperm)$ be an embedding and let $e$ be an edge of $\Emb$ with flags
$$x,\; x' = \Vperm(x),\; y = \Eperm(x),\; y' = \Vperm(y) = \Eperm(x').$$
We define an embedding $\Emb - e$, the result of deletion of the edge $e$ from $\Emb$, in the following manner.

First, we disconnect the edge $e$ from the rest of the embedding. For this, we modify $\Fperm$ so that $\{x,x'\}$ and $\{y,y'\}$ become its orbits:
If $\{x,x'\}$ is not already an orbit of $\Fperm$, we set $a=\Fperm(x)$ and $a'=\Fperm(x')$, and we replace $\{x,a\}$ and $\{x',a'\}$ with $\{x,x'\}$ and $\{a,a'\}$.
Next, if $\{y,y'\}$ is not already an orbit of $\Fperm$, we set $b=\Fperm(y)$, $b'=\Fperm(y')$, and we replace $\{y,b\}$ and $\{y',b'\}$ with $\{y,y'\}$ and $\{b.b'\}$.
As the third and final step, we delete the flags of $e$ and the corresponding orbits of all three involutions.

% of the vertices incident with $e$, keeping their cyclic order intact.
% To define what happens with $\Fperm$, let us first define an operation of bypassing an orbit $\{a,b\} \subseteq \{x,x',y,y'\}$ of $\Eperm$ or $\Vperm$.
% If $\{a,b\}$ is an orbit of $\Fperm$ as well, we just delete it from $\Fperm$. Otherwise, we replace the orbits $\{a, \Fperm(a)\}$ and $\{b,\Fperm(b)\}$
% (which involve four distinct flags) with an orbit $\{\Fperm(a), \Fperm(b)\}$. Note that bypassing an orbit keeps $\Fperm$ an involution, while deleting the flags $\{a,b\}$ from its domain.
% 
% Since in edge deletion our goal is to keep vertices intact, we first bypass the orbit $\{x,x'\}$ and then $\{y,y'\}$.
% A direct check shows that the order of bypassing (i.e., which vertex is called $x$) does not matter for the outcome: edge deletion corresponds to deleting the flags of $e$ from the cycles
% of the vertices incident with $e$, keeping their cyclic order intact.

Note that the first two steps cannot be conveyed in parallel because $\phi(y)$ and $\phi(y')$ might be altered during the first step.
However, a direct check shows that $\{x,x'\}$ and $\{y,y'\}$ both become orbits of $\Fperm$ and that $\Fperm$ keeps being an involution.
What is more, the order of the first two steps (i.e., the arbitrary decision of processing $\{x,x'\}$ prior to $\{y,y'\}$) is irrelevant  
and edge deletion results in deleting the flags of $e$ from the cycles corresponding to the vertices incident with $e$.
This interpretation supports the following observation.
\begin{observation}\label{obs:edge-del-vertices-stay}
Let $\Emb_0=\Emb-e$ for an embedding $\Emb$ and an edge $e$.
For every vertex $v_0$ of $\Emb_0$, there exists a distinct vertex $v$ of $\Emb$, such that the cycle of $v_0$ is the cycle of $v$ with the flags of $e$ removed.
In the other direction, for every vertex $v$ of $\Emb$, either all flags of $v$ are contained in $e$, or there exists a vertex $v_0$ of $\Emb_0$ with the cycle equal to the cycle of $v$
with the flags of $e$ removed.
\end{observation}

Next, we describe how genus changes subject to edge deletion.
\begin{lemma}\label{lem:edge-del}
Let $e$ be an edge in the embedding $\Emb$, and let $\Emb_0 = \Emb - e$.
Then the genus of $\Emb_0$ is not larger than the genus of $\Emb$. Furthermore, if $e$ is incident to two faces or to a size-$2$ vertex,
then the genera of $\Emb_0$ and $\Emb$ are equal.
\end{lemma}
\begin{proof}
Let $\Emb = (\Flags, \Eperm, \Vperm, \Fperm)$ and $\Emb_0 = (\Flags_0, \Eperm_0, \Vperm_0, \Fperm_0)$.
Clearly,  $|\orb(\Edges_{\Emb_0})| = |\orb(\Edges_\Emb)|-1$.
We now investigate how the number of vertices, faces, and connected components can change while deleting an edge $e$.
Let $x_1,x_2 = \Vperm(x_1), y_1 = \Eperm(x_1), y_2 = \Eperm(x_2) = \Vperm(y_1)$ be the four flags of $e$.

For vertices, recall that the operation of deleting an edge almost preserves the set of vertices: every vertex $v_0$ of $\Emb_0$ originates from a vertex $v$ of $\Emb$ by deleting
the flags of $e$ from the cycle of $v$. Thus, $|\orb(\Vertices_\Emb)| - |\orb(\Vertices_{\Emb_0})|$ equals the number of vertices of $\Emb$ that have all flags contained in $e$.
There may be $0$, $1$, or $2$ of them. Furthermore, if $e$ is incident to two distinct faces, there are no such vertices, unless $e$ is a loop at a vertex $v$ incident only to $e$; in this case,
$e$ is incident to two faces of size $2$, and its set of flags is a whole connected component of $\Emb$.

For faces, consider first the case when $e$ is incident to two faces, $f_1$ and $f_2$, and assume that $x_i,y_i$ lie on $f_i$ for $i=1,2$.
Then, if the cycle of $f_i$ is $x_i, y_i, P_i$ for some sequence of flags $P_i$, then the deletion of $e$ replaces the faces $f_1$ and $f_2$ with one face with cycle $P_1 \bar{P}_2$, where
$\bar{P}_2$ is the sequence $P_2$ reversed. Consequently, $|\orb(\Faces_\Emb)| - |\orb(\Faces_{\Emb_0})| = 1$ in this case.

Consider now a case when $e$ is incident to one face $f$. If the cycle of $f$ is $x_1, y_1, P_1, x_2, y_2, P_2$ for some sequences of flags $P_1$ and $P_2$, then the deletion of 
$e$ replaces $f$ with a face with a cycle $P_1 \bar{P}_2$. If the cycle of $f$ is $x_1, y_1, P_1, y_2, x_2, P_2$ for some $P_1,P_2$, then the deletion of $e$
replaces $f$ with two faces with cycles $P_1$ and $P_2$, respectively. 
Consequently, in this case we have $|\orb(\Faces_\Emb)| - |\orb(\Faces_{\Emb_0})| \in \{0, -1\}$.

For connected components, note that $|\orb(\conncomp_\Emb)| - |\orb(\conncomp_{\Emb_0})| \in \{-1,0,1\}$: either $e$ connects two distinct vertices that land in different connected components of $\Emb_0$,
or the set of connected components essentially does not change (one connected component loses the flags of $e$) or $e$ is a whole connected component and it disappears in $\Emb_0$.
    
Let us now wrap up the argument. If the set of flags of $e$ is a whole connected component, then $|\orb(\conncomp_\Emb)| - |\orb(\conncomp_{\Emb_0})| = 1$.
This may happen if $e$ is a loop at a vertex incident only to $e$, or if $e$ connects two vertices incident only to $e$.
In the first case, let $k \in \{1,2\}$ be the number of faces $e$ is incident with. We have
\begin{align*}
|\orb(\Faces_\Emb)-|\orb(\Faces_{\Emb_0})| &= k, & |\orb(\Vertices_\Emb)| - |\orb(\Vertices_{\Emb_0})| &= 1, & |\orb(\conncomp_\Emb)| - |\orb(\conncomp_{\Emb_0})| &= 1.
\end{align*}
Consequently, $\genus_\Emb - \genus_{\Emb_0} = 2-k \in \{0, 1\}$. In particular, the genera of $\Emb$ and $\Emb_0$ are equal if $e$ is incident to two faces.

In the second case, when $e$ connects two size-$2$ vertices, $e$ is incident with one face with the same set of flags. Thus
\begin{align*}
|\orb(\Faces_\Emb)|-|\orb(\Faces_{\Emb_0})| &= 1, & |\orb(\Vertices_\Emb)| - |\orb(\Vertices_{\Emb_0})| &= 2, & |\orb(\conncomp_\Emb)| - |\orb(\conncomp_{\Emb_0})| &= 1.
\end{align*}
Consequently, $\genus_\Emb - \genus_{\Emb_0} = 0$.

Now consider a case when $e$ is incident to two faces, but is not a whole connected component. Then, at least one of this faces is of size larger than $2$, say $f_1$ with cycle $x_1,y_1,P_1$.
Then, $P_1$ remains in $\Emb_0$ as a sequence of flags, where every two consecutive flags form an orbit either of $\Fperm$ or of $\Eperm$. 
Consequently, the flags of the vertices incident with $e$
remain in the same vertices and in the same connected component, as $P_1$ connects the first and the last flags of $P_1$. We infer that in this case we have:
\begin{align*}
|\orb(\Faces_\Emb)|-|\orb(\Faces_{\Emb_0})| &= 1, & |\orb(\Vertices_\Emb)| - |\orb(\Vertices_{\Emb_0})| &= 0, & |\orb(\conncomp_\Emb)| - |\orb(\conncomp_{\Emb_0})| &= 0.
\end{align*}
Consequently, $\genus_\Emb = \genus_{\Emb_0}$. Note that at this point we have concluded the proof that $\genus_\Emb = \genus_{\Emb_0}$ if $e$ is incident to two faces.

Consider now a case when $e$ is incident to one face, and there exists its incident vertex $v$ with the set of flags contained in $e$. As we have already excluded the case
when the flags of $e$ is a whole connected component, $e$ is incident with two vertices $v$ and $v'$, and the flags of $v'$ are not contained in $e$.
Furthermore, $v$ is of size $2$, say $v$ consists of flags $y_1$ and $y_2$. Then, the cycle of the face incident with $e$ is $x_1,y_1,y_2,x_2,P$ for some sequence of flags $P$.
Consequently,
\begin{align*}
|\orb(\Faces_\Emb)|-|\orb(\Faces_{\Emb_0})| &= 0, & |\orb(\Vertices_\Emb)| - |\orb(\Vertices_{\Emb_0})| &= 1, & |\orb(\conncomp_\Emb)| - |\orb(\conncomp_{\Emb_0})| &= 0.
\end{align*}
We infer that $\genus_\Emb = \genus_{\Emb_0}$. Note that at this point we have concluded the proof that $\genus_\Emb = \genus_{\Emb_0}$ if $e$ is incident to a size-$2$ vertex. 

Assume now that $e$ is incident to one face, and every vertex incident with $v$ has its set of flags not contained in $e$. 
Let $f$ be the face incident with $e$. If the cycle of $f$ is $x_1,y_1,P_1,x_2,y_2,P_2$, then the resulting face with cycle $P_1 \bar{P}_2$ provides connectivity in $\Emb_0$
between the (remainings of) vertices incident with $e$. Consequently, 
In this case
\begin{align*}
|\orb(\Faces_\Emb)|-|\orb(\Faces_{\Emb_0})| &= 0, & |\orb(\Vertices_\Emb)| - |\orb(\Vertices_{\Emb_0})| &= 0, & |\orb(\conncomp_\Emb)| - |\orb(\conncomp_{\Emb_0})| &= 0,
\end{align*}
and have then $\genus_\Emb - \genus_{\Emb_0} = 1$.

In the last case, if the cycle of $f$ is $x_1,y_1,P_1,y_2,x_2,P_2$, we have
\begin{multline*}
|\orb(\Faces_\Emb)|-|\orb(\Faces_{\Emb_0})| = -1, \quad |\orb(\Vertices_\Emb)| - |\orb(\Vertices_{\Emb_0})| = 0,\\ |\orb(\conncomp_\Emb)| - |\orb(\conncomp_{\Emb_0})| \in \{0,-1\},
\end{multline*}
and then $\genus_\Emb - \genus_{\Emb_0} \in \{0, 2\}$.
\end{proof}

A direct corollary of Lemma~\ref{lem:edge-del} is the following.

\begin{corollary}\label{cor:genus-nonneg}
The genus of an embedding is always nonnegative.
\end{corollary}
\begin{proof}
Observe that any embedding can be turned into an empty embedding (i.e., one with no flags) by successive edge deletions. 
Since edge deletion cannot increase the genus (Lemma~\ref{lem:edge-del}), and the empty embedding has genus zero, the lemma follows.
\end{proof}

\subsubsection{Drawing a new edge}
Let us now define a reverse operation to edge deletion.
Before, let us introduce a notion of \emph{position} $p_\Emb(x)$ of a flag $x$. 
Let $x,y=\Eperm(x), y'=\Vperm(y), x'=\Vperm(x)$ be the flags contained in the edge $e$ containing $x$. 
We set $p_\Emb(x)=(\Fperm(x),\Fperm(y))$ if $\Fperm(y)\notin \{x,x'\}$ and $p_\Emb(x)=(\Fperm(x),\Fperm(\Vperm(\Fperm(y))))$ otherwise. 
It is easy to observe that  defining the process of deletion of $e$, we set $a$ and $b$ so that $(a,b)=p_\Emb(x)$.
This lets us use $p_\Emb(x)$ to undo the deletion.
We often do not want to specify the flags of the new edge. In this case, we write $a=\bot$ instead of $a=x'$, $a=\top$ instead of $a=y$,
$a=\top'$ instead of $a=y'$, and $b=\bot$ instead of $b=y'$.
% 
%  which we define as $p_\Emb(x)=(a,b)$
% where $a=\phi(x)$ and $b=\phi()$

Let $\Emb = (\Flags, \Eperm, \Vperm, \Fperm)$ be an embedding, and let $a\in \Flags\cup\{\bot,\top,\top'\}$ and $b\in \Flags\cup\{\bot\}$.
By \emph{drawing a new edge $(x,y,y',x')$ at $(a,b)$} we mean the following operation, resulting in an embedding $\Emb_0$. 
% 
% 
% let $\{x,y,y',x'\}\notin \Flags$,
%  $a\in \Flags\cup\{x',y,y'\}$, and $b\in \Flags\cup\{y'\}$.
%  
%  and let $a,b \in \Flags \cup \{\bot,\top,\bot'\}$, with the additional assumption that if at least one of $a$ or $b$ is not in $\Flags$,
% then $a = \bot$ or $b = \bot$.
% If $a \in \Flags$, we define $a' = \Fperm(a)$, and similarly $b' = \Fperm(b)$ if $b \in \Flags$.
First, we add new flags $\{x,y,y',x'\}$ to $\Flags$ and we set $\{x,x'\}$ and $\{y,y'\}$ to be new orbits of $\Vperm$ and $\Fperm$,
and setting $\{x,y\}$ and $\{x',y'\}$ to be two new orbits of $\Eperm$.
Next, we replace $a=\bot$, $a=\top$, $a=\top'$, and $b=\bot$ with $a=x'$, $a=y$, $a=y'$, and $b=y'$, respectively.
Finally, we adjust $\Fperm$ in the following two steps: If $b\ne y'$, we define $b'=\Fperm(b)$ and replace the orbits $\{y,y'\}$ and $\{b,b'\}$ with $\{y,b\}$ and $\{y',b'\}$.
Then, if $a \ne x'$, we define $a'=\Fperm(a)$ and the orbits $\{x,x'\}$ and $\{a,a'\}$ with $\{x,a\}$ and $\{x',a'\}$.

We defined this operation so that it is clear that if $e=(x,y,y',x')$ is an edge of $\Emb$, then one can retrieve $\Emb$ from $\Emb-e$ by drawing a new edge $(x,y,y',x')$ at $p_\Emb(x)$. 
It is also easy to verify that drawing a new edge is a well defined and that if $\Emb_0$ is obtained from $\Emb$ by drawing a new edge $e$, then $\Emb=\Emb_0-e$.

 Note that $a=\bot$ and $b=\bot$ both result in creating a new vertex
(two new vertices if $a=\bot$ and $b=\bot$ hold simultaneously). Moreover, if $a=\bot$, then $x$ belongs to a new size-2 vertex $(x,x')$, 
while $a=\top$, then $x$ belongs to a new size-2 face $(x,y)$.

We identify one more special case of drawing a new edge. If $a$ and $b$ are distinct flags that lie on the same face $f$, and furthermore, the order of flags on the cycle of $f$ is $\Fperm(a), a, P, b, \Fperm(b), P'$
for some (possibly empty) sequences of flags $P$ and $P'$, then we say that the new edge \emph{is drawn along the boundary of $f$}. 
Observe that if this is the case, then the new edge $e$ is
incident to two faces: in the new embedding $\Emb_0$, the face $f$ has been split
into a face with cycle $P, a, x, y, b$ and a face with cycle $P', \Fperm(b), y', x', \Fperm(a)$.
We explicitly allow here also the case $b = \Fperm(a)$; then the cycle of $f$ is $a, P, b$ for some sequence $P$, and the new embedding has new faces with cycles $x',y'$ and $P, a, x, y, b$.

Consequently, from Lemma~\ref{lem:edge-del} we immediately obtain the following.
\begin{lemma}\label{lem:edge-draw}
If $\Emb_0$ is created from $\Emb$ by drawing a new edge, then the genus of $\Emb_0$ is not smaller than the genus of $\Emb$.
Furthermore, the genera of $\Emb$ and $\Emb_0$ are equal if the edge has been drawn along a face boundary or at $(a,b)$ with  $a\in \{\bot,\top\}$.
\end{lemma}

\subsection{$t$-boundaried embeddings}

\begin{definition}[$t$-boundaried embedding]
A \emph{$t$-boundaried embedding} is a tuple $(\Emb, t, \Label)$ where $\Emb = (\Flags, \Eperm, \Vperm, \Fperm)$ is an
embedding, $t$ is a nonnegative integer, and $\Label$ is an injective function from a subset of $\orb(\Vertices_\Emb)$ to $[t]$.
The elements of the domain of $\Label$ are called \emph{labelled vertices}, and the elements of $[t]$ are \emph{labels}.

A genus of a $t$-boundaried embedding $(\Emb, t, \Label)$ is the genus of the underlying embedding $\Emb$.
\end{definition}

The main motivation to introduce $t$-boundaried embeddings is to then merge them.

\begin{definition}[merge of two $t$-boundaried embeddings]
Let $\tEmb_1$ and $\tEmb_2$ be two $t$-boundaried embeddings, where $\tEmb_i = (\Emb_i, t, \Label_i)$ and $\Emb_i = (\Flags_i, \Eperm_i, \Vperm_i, \Fperm_i)$ for $i=1,2$,
    and the sets of flags $\Flags_1$ and $\Flags_2$ are disjoint.
A $t$-boundaried embedding $\tEmb = (\Emb, t, \Label)$ with $\Emb = (\Flags, \Eperm, \Vperm, \Fperm)$ is a \emph{merge} of $\tEmb_1$ and $\tEmb_2$ if it is created by the following process.

First, we take $\tEmb$ to be a disjoint union of $\tEmb_1$ and $\tEmb_2$, that is, we take:
\begin{align*}
\Flags & = \Flags_1 \cup \Flags_2 &
\Eperm & = \Eperm_1 \cup \Eperm_2 \\
\Vperm & = \Vperm_1 \cup \Vperm_2 & 
\Fperm & = \Fperm_1 \cup \Fperm_2 \\
\Label &= \Label_1 \cup \Label_2 & &
\end{align*}
Then, for every label $\ell \in [t]$ that is contained in both the range of $\Label_1$ and $\Label_2$, that is, there exists a cycle $C_i$ corresponding to the orbit $\Label_i^{-1}(\ell)$ for $i=1,2$,
we modify $\Fperm$ on the elements of $C_1$ and $C_2$ so that these elements form a single orbit of $\subgroup{\Vperm, \Fperm}$ whose cycle $C$ is a merge of $C_1$ and $C_2$.
We assign this cycle the label $\ell$ in the assignment $\Label$.

A \emph{genus-minimum} merge of $\tEmb_1$ and $\tEmb_2$ is a merge that minimizes its genus.
\end{definition}
Note that, as we only modify $\Fperm$ in the merge operation, the set of edges of a merge is a union of the edges of the components. In particular, it follows that every edge of a merge consists
of four flags, and thus it is indeed an embedding.

\begin{definition}[equivalence of $t$-boundaried embeddings]
Two $t$-boundaried embeddings $\tEmb_1$ and $\tEmb_2$ are \emph{equivalent} if
for every $t$-boundaried embedding $\tEmb$, the genera of genus-minimum merges of $\tEmb$ and $\tEmb_i$ for $i=1,2$ are equal.
\end{definition}

By Observation~\ref{obs:edge-del-vertices-stay}, there is a natural correspondence between the vertices of $\Emb$ and $\Emb-e$ for every edge $e$ of $\Emb$. 
This correspondence allows us to extend the definition of edge deletion to $t$-boundaried embeddings: if $\tEmb = (\Emb, t, \Label)$ and $e$ is an edge in $\Emb$,
then a $t$-boundaried embedding $\tEmb-e$ is defined as $(\Emb-e, t, \Label_0)$, where $L_0$ is defined so that for every
vertex $v_0$ of $\Emb-e$, we take the corresponding vertex $v$ of $\Emb$, and copy its label to $v_0$ if $v$ is labelled.
Note that the range of $\Label_0$ may be a proper subset of the range of $\Label$ if some labelled vertex in $\tEmb$ has all its flags contained in the deleted edge $e$.

Furthermore, the characterization of vertex cycles in \cref{obs:edge-del-vertices-stay} lets us relate edge deletion to merging.
\begin{observation}\label{obs:merge-deletion}
Consider a merge $\tEmb$ of $t$-boundaried embedding $\tEmb_1$ and $\tEmb_2$.
For every edge $e$ of $\tEmb_1$, we have that $\tEmb-e$ is a merge of $\tEmb_1-e$ and $\tEmb_2$.
\end{observation}

Similarly, we can define the operation of drawing a new edge in a $t$-boundaried embedding; if a new vertex is created by this operation (due to $a$ or $b$ being equal to $\bot$),
we need to specify its label (or the fact that it is unlabelled). Thus, for each label $\ell$ unused in $\tEmb$, we add a special value $\bot_\ell$ available for $a$ and $b$, denoting the fact that 
the new vertex remains is labelled $\ell$; ordinary $\bot$ denotes the fact that it is unlabelled.
A counterpart of \cref{obs:merge-deletion} requires dealing with a special situation.

\begin{observation}\label{obs:merge-draw}
Consider a merge $\tEmb$ of $t$-boundaried embedding $\tEmb_1$ and $\tEmb_2$,
and let $\tEmb_1'$ be an embedding obtained from $\tEmb_1$ by drawing an edge at $(a,b)$.
Then a merge of $\tEmb_1'$ and $\tEmb_2$ can be obtained from $\tEmb$ by drawing an edge at $(a',b')$ ,
where $a'=a$ and $b'=b$ except for the following situation:
   if $a=\bot_\ell$  (or $b=\bot_\ell$) for a label $\ell$ used in $\tEmb$ but not in $\tEmb_1$, then $a'$ (resp. $b'$) is an arbitrary flag of $\tEmb$ contained in a vertex with label $\ell$.
\end{observation}

\subsection{Nice embeddings}
We say that a vertex  or a face is \emph{isolated} if its set of flags is a whole connected component. 
\begin{definition}[nice ($t$-boundaried) embedding]
A $t$-boundaried embedding $\tEmb = (\Emb, t, \Label)$ with $\Emb = (\Flags, \Eperm, \Vperm, \Fperm)$ is \emph{nice} if the following two conditions hold:
\begin{enumerate}
%\item if an edge is incident to two vertices, it is incident to at least one labelled vertex;
\item If an unlabelled vertex consists of less than $6$ flags, then it is isolated.
\item If an edge $e$ is incident to two faces, then for every face $f$ incident with $e$, if we denote by $x$ and $y = \Eperm(x)$ the two flags contained both in $f$ and $e$,
  then there is a flag $z\in f\setminus \{x,y,\Fperm(x),\Fperm(y)\}$ that is contained in a labelled vertex.
\end{enumerate}
\end{definition}

Our goal in this section is to show that any embedding can be turned into an equivalent nice one.
The main motivation for such a cleaning step is that a nice embedding enjoys a good size bound due to the Euler formula-style estimations, presented in the next section.

\subsubsection{Nice embeddings are small}

\begin{lemma}\label{l:nice-bound}
A $t$-boundaried nice embedding $\tEmb = (\Emb, t, \Label)$ of genus $\genus$ satisfies
$$|\Flags(\Emb)| \leq 48t+24\genus_\tEmb.$$
\end{lemma}
\begin{proof}
Let $\Emb = (\Flags, \Eperm, \Vperm, \Fperm)$. We perform a discharging argument. 
The setup is as follows:
\begin{itemize}
\item every labelled vertex receives a charge of $2$;
\item every isolated vertex receives a charge of 1;
\item every isolated face receives a charge of 1;
\item every edge receives a charge of $\frac56$.
\end{itemize}
The total initial charge is at most
$$2t+2|\orb(\orb(\conncomp_\Emb))|+\tfrac56|orb(\Edges_\Emb)|.$$% = 2t+2|\orb(\conncomp_\Emb)| + 19/80|\Flags|.$$

Then we move the charge according to the following rules:
\begin{enumerate}
\item Every labelled vertex that is incident to only one face, sends a charge of $1$ to the face it is incident with.
\item Every edge that is incident only to labelled vertices, divides a charge of $\frac23$ equally among the faces it is incident with.
In other words, every flag in such an edge sends a charge of $\frac16$ to the face it is contained in.
\item Every edge that is incident with two vertices, one labelled and one unlabelled, gives a charge of $\frac13$ to the unlabelled incident vertex, and divides the remaining charge of $\frac12$ equally between
the faces it is incident with.
In other words, in the first part every flag in such an edge contained in unlabelled vertex sends a charge of $\frac16$ to the vertex it is contained in.
In the second part, every flag in such an edge sends a charge of $\frac18$ to the face it is contained in.
\item Every edge that is incident only to unlabelled vertices, divides the charge of $\frac23$ equally among the vertices it is incident with.
In other words, every flag in such an edge sends a charge of $\frac16$ to the vertex it is contained in.
\end{enumerate}
Clearly, every edge is left with non-negative charge ($0$ or $\frac16$). We now show that at the end of the process, every vertex and every face has charge of at least one.

First, let us consider a vertex $v$. If $v$ is labelled or isolated, $1$ out of its initial charge remained, and the claim is straightforward. 
Next, we assume that $v$ is unlabelled and not isolated. 
By the first property of a nice embedding, $v$ is of size at least $6$. 
It received exactly $\frac16$ from each of its flags, i.e., at least 1 in total.

Consider now a face $f$. We make case distinction depending on how many flags of $f$ belong to labelled vertices
and how these flags are located on the cycle corresponding to $f$.

\medskip

 \noindent\textbf{No flags of $f$ belong to a labelled vertex.} 
 Observe that every edge $e$ incident with $f$ is incident only with one face, as otherwise it would contradict the second property of a nice embedding.
Consequently, $f$ is isolated an it received an initial charge of 1.

\noindent\textbf{Exactly 2 flags of $f$ belong to a labelled vertex.} Suppose $f$ is incident with only one labelled vertex $v$ and shares two flags $x$ and $y = \Fperm(x)$ with $v$.
Let $e$ be the edge containing $x$. By the second property of a nice embedding, $f$ is the only face incident with $e$.
Consequently, $\Vperm(x)$ belongs to $f$. Since $\Vperm(x)$ belongs to $v$, which is a labelled vertex, we infer that $y = \Vperm(x)$, the vertex $v$ is of size $2$, and $f$ is the only face
$v$ is incident with. Thus, $f$ received a charge of $1$ from $v$.

\noindent\textbf{Exactly $4$ flags of $f$ belong to labelled vertices and they are consecutive along $f$.}
Let $x$, $y = \Eperm(x)$, $x' = \Fperm(x)$, and $y' = \Fperm(y)$ be these four flags, and let $e$ be the edge containing $x$ and $y$.
If $e$ is incident to two faces, then it violates the last property of a nice embedding. 
Otherwise, the orbit $\{\Vperm(x),\Vperm(y)\}$ of $\Eperm$ appears on $f$. Since only $4$ flags of $f$ belong to labelled vertices, we need to have $\{\Vperm(x),\Vperm(y)\} = \{x',y'\}$,
  and $f$ is an isolated face of size $4$. Hence, it received an initial charge of 1.
  
\noindent\textbf{The face $f$ consists of $6$ flags and they all belong to labelled vertices.}
In this case all flags in $f$ are contained in edges incident only to labelled vertices. 
Thus, $f$ receives a charge of $\frac16$ from each of these flags, which is $1$ in total. 

\noindent\textbf{Along $f$ there exist two nonconsecutive orbits of $\Fperm$ that are included in labelled vertices.}
Let $\{x,x'\}$ and $\{y,y'\}$ be these orbits. Since they are nonconsecutive, $x,x',\Eperm(x),\Eperm(x'),y,y',\Eperm(y),\Eperm(y')$ are eight pairwise distinct flags in $f$,
each sending a charge of at least $\frac18$ to the face $f$.

Note that the last three cases case cover the case of $f$ having at least $4$ flags contained in labelled vertices.

\medskip

We have shown that there was enough charge so that every vertex and every face received a charge of at least one. Consequently,
$$|\orb(\Faces_\Emb)| + |\orb(\Vertices_\Emb)| \leq 2t+2|\orb(\conncomp_\Emb)|+ \tfrac56|\orb(\Edges_\Emb)|.$$
Together with~\eqref{eq:genus}, it implies that
$$\genus_{\tEmb}=|\orb(\Edges_\Emb)|-|\orb(\Faces_\Emb)|-|\orb(\Vertices_\Emb)|+2|\orb(\conncomp_\Emb)| \geq \tfrac16|\orb(\Edges_\Emb)|-2t,$$
i.e.,
$$|\Flags|=4|\orb(\Edges_\Emb)| \leq 48t + 24\genus_\tEmb.$$
\end{proof}

\begin{corollary}\label{cor:few-embeddings}
There are $2^{\Oh((t+\genus) \log (t+\genus))}$ nice $t$-boundaried embeddings of genus at most $\genus$,
and they can be enumerated in time $2^{\Oh((t+\genus) \log (t+\genus))}$.
\end{corollary}

\subsubsection{Making an embedding nice}

Let $\tEmb = (\Emb, t, \Label)$ be a $t$-boundaried embedding with $\Emb = (\Flags, \Eperm, \Vperm, \Fperm)$.
Our goal now is to obtain an equivalent nice embedding. 
To this end, we show that the following three operations lead to equivalent embeddings:
\begin{enumerate}
\item deleting an edge incident with an unlabelled size-$2$ vertex;
\item deleting an edge violating the last property of the definition of a nice embedding;
\item suppressing a size-$4$ unlabelled vertex that is not isolated.
\end{enumerate}
We will henceforth call them \emph{simplifying operations}.

\begin{lemma}[deleting an edge incident to a size-$2$ unlabelled vertex]\label{lem:del-e-size2v}
If $e$ is an edge of $\tEmb$ incident to an unlabelled vertex of size $2$,
then $\tEmb-e$ and $\tEmb$ are equivalent.
\end{lemma}
\begin{proof}
We consider merges $\tEmb$ and $\tEmb-e$ with an arbitrary $t$-boundaried embedding $\tEmb_1$.  

First, let $\tEmb_M$ be a merge of $\tEmb_1$ and $\tEmb$. 
\cref{obs:merge-deletion} yields that $\tEmb_M-e$ is a merge of $\tEmb-e$ and $\tEmb_1$
and \cref{lem:edge-del} implies that its genus does not exceed the genus of $\tEmb_M$.
Consequently, the genus of a genus-minimum merge of $\tEmb_1$ and $\tEmb-e$ is not larger than the genus of a genus-minimum merge of $\tEmb_1$ and $\tEmb$.

In the other direction, let $\tEmb_M'$ be a merge of $\tEmb-e$ and $\tEmb_1$.
Let $e=(x,y,y',x')$ where $\{x,x'\}$ is contained in an unlabelled size-2 vertex so that $p_\tEmb(x)$ is of the form $(x',b)$.
Consequently, $\tEmb$ can be obtained from $\tEmb-e$ by drawing a new edge $e$ at $(\bot,b)$ for some $b$.
By \cref{obs:merge-draw}, a merge $\tEmb_M$ of $\tEmb_1$ and $\tEmb$ can be obtained from $\tEmb_M'$ by drawing a new edge at $(a',b')$.
Since $a=\bot$, we have $a'=\bot$, so the genera of $\tEmb_M$ and $\tEmb_M'$ are equal due to \cref{lem:edge-draw}.
Consequently, the genus of a genus-minimum merge of $\tEmb_1$ and $\tEmb$ is not larger than the genus of a genus-minimum merge of $\tEmb_1$ and $\tEmb-e$.
This completes the proof of the lemma.
\end{proof}

\begin{lemma}[deleting an edge violating the last property of the definition of nice embedding]\label{lem:del-e-empty-f}
Let $e$ be an edge in a $t$-boundaried embedding $\tEmb$ that is incident to two faces.
Furthermore, assume that one face $f$ incident with $e$ has the following property: if $x$ and $y = \Eperm(x)$ are the two flags shared between $e$ and $f$, then
each flag $z \notin \{x,\Fperm(x),y,\Fperm(y)\}$ on $f$ belongs to an unlabelled vertex.
Then $\tEmb$ is equivalent with $\tEmb-e$.
\end{lemma}
\begin{proof}
Again, we consider merges of $\tEmb$ and $\tEmb-e$ with an arbitrary $t$-boundaried embedding $\tEmb_1$.  

In one direction, the proof is the same as in the proof of the previous lemma:
\cref{obs:merge-deletion} and \cref{lem:edge-del} imply that the genus of a genus-minimum merge of $\tEmb_1$ and $\tEmb-e$ is not larger than the genus of a genus-minimum merge of $\tEmb_1$ and $\tEmb$.

In the other direction, let $\tEmb_M'$ be a merge of $\tEmb_1$ and $\tEmb-e$.
We consider two cases depending on whether $f$ is of size 2.
If so, then $p_{\Emb(x)}$ is of the form $(y,b)$, so $\tEmb$ can be obtained from $\tEmb-e$ by drawing an new edge at $(\top, b)$ for some $b$.
By \cref{obs:merge-draw}, a merge $\tEmb_M$ of $\tEmb_1$ and $\tEmb$ can be obtained from $\tEmb_M'$ by drawing a new edge at $(a',b')$.
Since $a=\top$, we have $a'=\top$, so the genera of $\tEmb_M$ and $\tEmb_M'$ are equal due to \cref{lem:edge-draw}.

Next, suppose that $f$ contains a flag $z\notin \{x,y\}$. 
Since $e$ is incident to 2 faces, we conclude that $\Vperm(x)$ and $\Vperm(y)$ do not belong to $e$,
and $\tEmb$ can be obtained from $\tEmb-e$ by drawing a new edge at $(a,b)$ for $a=\Vperm(x)$ and $b=\Vperm(y)$.
By \cref{obs:merge-draw}, a merge $\tEmb_M$ of $\tEmb$ and $\tEmb_1$ can be obtained from $\tEmb_M'$ by drawing a new edge at $(a,b)$.
Let the cycle of $f$ be $x,a,P,b,y$ for some (possibly empty) sequence of flags $P$. By the assumptions of the lemma, every flag of $P$ belongs to an unlabelled vertex.
Consequently, in $\tEmb_M'$ there exists a face $f'$ whose cycle contains consecutive flags $a, P, b$ on its cycle, as no orbit of $\Fperm$ or $\Eperm$ on $a, P, b$ has been altered
This means that a new edge drawn at $(a,b)$ is actually drawn at a face boundary. Hence, the genera of $\tEmb_M$ and $\tEmb_M'$ are equal due to \cref{lem:edge-draw}. 

Consequently, the genus of a genus-minimum merge of $\tEmb_1$ and $\tEmb$ is not larger than the genus of a genus-minimum merge of $\tEmb_1$ and $\tEmb-e$,
which concludes the proof of the lemma.
\end{proof}

For the last basic operation, we need to formally define it.
Let $\tEmb = (\Emb, t, \Label)$ be a $t$-boundaried embedding with an unlabelled size-$4$ vertex $v$ that is not isolated.
Since $v$ is not isolated, $v$ is incident with two edges $e_1$ and $e_2$, and each $e_i$ is incident with $v$ and a vertex $v_i \neq v$. Note that it is possible that $v_1 = v_2$.

Let $x_1$ and $x_2 = \Fperm(x_1)$ be two flags in $v$ such that $x_i$ belongs to $e_i$. Furthermore, let $y_i = \Eperm(x_i)$; note that $y_i$ lies in $e_i$.
We delete all four flags of $v$, and replace the orbits of $\Eperm$ on $e_1$ and $e_2$ with $\{y_1, y_2\}$ and $\{\Vperm(y_1), \Vperm(y_2)\}$.
Clearly, we have replaced $e_1$ and $e_2$ with a new edge with flags $y_1, y_2, \Vperm(y_1), \Vperm(y_2)$. We now formally verify that the output embedding
is an equivalent one.

\begin{lemma}[suppressing a size-$4$ unlabelled vertex that is not isolated]\label{lem:supp-4}
The operation of suppressing a size-$4$ unlabelled vertex leads to an equivalent embedding.
\end{lemma}
\begin{proof}
We interpret the suppressing operation as a sequence of one edge drawing and two edge deletions.

Observe that $y_1,x_1,x_2,y_2$ are four distinct flags that lie on the same face $f$ and, furthermore, they are consecutive in this order along the cycle of $f$.
Let us draw a new edge $e_f$ along the boundary of $f$, at $y_1$ and $y_2$, obtaining an embedding $\tEmb_f$.
Note that in $\tEmb_f$ there is a new face $f'$ with cycle $y_1,x_1,x_2,y_2,z_2,z_1$, where $\{z_1,z_2\}$ is a new orbit of $\Eperm$ contained in the edge $e_f$.
Furthermore, the edge $e_f$ with face $f'$ fulfills the assumptions of Lemma~\ref{lem:del-e-empty-f}, and its deletion from $\tEmb_f$ gives $\tEmb$.
Consequently, $\tEmb_f$ and $\tEmb$ are equivalent.

Then, we delete edges $e_1$ and $e_2$ from $\tEmb_f$, obtaining an embedding $\tEmb'$. Our goal is to show that such an operation leads to an equivalent embedding;
note that if we rename the flags of $e_f$ to $y_1,y_2,\Vperm(y_1),\Vperm(y_2)$ we obtain the output embedding of the suppressing operation.
The proof is similar to that of Lemma~\ref{lem:del-e-empty-f}, but without most of the special cases due to the existence of the edge $e_f$.

Let $\tEmb_1$ be an arbitrary $t$-boundaried embedding.
By \cref{obs:merge-deletion} and \cref{lem:edge-del} yields that the genus of a genus-minimum merge of $\tEmb_1$ and $\tEmb'$ is not larger than the genus of a genus-minimum merge of $\tEmb_1$ and $\tEmb_f$.

For the other direction, let $\tEmb_M' = (\Emb_M', t, \Label_M')$ be a merge of $\tEmb_1$ and $\tEmb'$.
Note that $p_{\tEmb_f}(x_1)=(x_2,z_1)$ and $p_{\tEmb_f-e_1}(x_2)=(x_2',z_2)$.
Hence, $\tEmb_f$ can be retrieved from $\tEmb'$ by drawing edge $e_2$ at $(\bot,z_2)$ and then edge $e_1$ at $(x_2,z_1)$.
By \cref{obs:merge-draw}, a merge $\tEmb_M$ of $\tEmb_f$ and $\tEmb_1$ can be retrieved from $\tEmb_M'$ by drawing an edge $e_2$
at $(\bot,z_2)$ and then $e_1$ at $(x_2,z_1)$. \cref{lem:edge-draw} yields that the first of these operations asserts that the genus remains unchanged.
As for the second operation, we observe after drawing $e_2$, we have that $x_2, y_2, z_2,z_1$ is a fragment of face boundary and that $e_1$ is drawn along it.
Hence, \cref{lem:edge-draw} also implies that the genera of $\tEmb_M'$ and $\tEmb_M$ are equal.
We infer that the genus of a genus-minimum merge of $\tEmb_1$ and $\tEmb_f$ is not larger than the genus of a genus-minimum merge of $\tEmb_1$ and $\tEmb'$.
This concludes the proof of the lemma.
\end{proof}

Armed with the three basic operations, we are now ready to prove the following statement.
\begin{lemma}\label{lem:make-nice}
There exists a polynomial-time algorithm that, given an arbitrary $t$-boundaried embedding, computes an equivalent nice one
by successive applications of the simplifying operations.
\end{lemma}
\begin{proof}
Let $\tEmb$ be a $t$-boundaried embedding. 
If $\tEmb$ is nice, we can just return $\tEmb$. Otherwise, observe that any object that violates the niceness of $\tEmb$
gives rise to an application of one of the simplifying operations.
An edge violating the last property of a nice embedding can be deleted from $\tEmb$; the resulting embedding is equivalent due to Lemma~\ref{lem:del-e-empty-f}.
Similarly, if there exists an unlabelled size-$2$ vertex, then Lemma~\ref{lem:del-e-size2v} allows us to safely delete the incident edge,
and a size-$4$ unlabelled vertex that is not isolated can be also suppressed by Lemma~\ref{lem:supp-4}.

Finally, note that each execution of a simplifying operation strictly decreases the number of flags in the embedding. 
Consequently, in polynomial time we obtain an equivalent nice embedding.
\end{proof}

\section{Bounded treewidth graphs}\label{sec:tw}

In this section we prove Theorem~\ref{thm:genusVDtw}. Without loss of generality, we can assume that the input tree decomposition $(T,\tdbag)$ of the input graph $G$
is a nice tree decomposition of width \emph{less} than $k$, that is, every bag of $(T,\tdbag)$ is of size at most $k$.

We first refine $(T,\beta)$ as follows. First, with every node $t$ we associate a graph $G(t)$, initialized as $G(t) := \tdG(t)$. 
In the refinement process, we will maintain the invariant that at every node $t$, the graph $G(t)$ is a subgraph of $G[\tdunder(t)]$ and a supergraph 
of $\tdG(t)$; in particular, $V(G(t)) = V(\tdG(t))$.
Note that one needs only $\Oh(t^2)$ bits to store $G(t)$ at every node $t$, as one needs only to store which edges of $G[\tdbag(t)]$ belong to $G(t)$.

For every forget node $t$ we perform the following refinement process. Let $t'$ be the child of $t$, and let $v$ be the forgotten vertex, i.e., $\{v\} = \tdbag(t') \setminus \tdbag(t)$. 
Let $e_1,e_2,\ldots,e_\ell$ be the edges incident with $v$ that have their second endpoint in $\tdbag(t)$; note that $\{e_1,e_2,\ldots,e_\ell\}$
are exactly the edges that are present in $\tdG(t)$ but are not present in $\tdG(t')$. 
Our goal is to refine the edge $tt'$ of $T$ by inserting a number of vertices on this edge so that the transition from $\tdG(t')$ to $\tdG(t)$ is more smooth.

More precisely, we subdivide the edge $tt'$ $\ell$ times, inserting vertices $t_1,t_2,\ldots,t_\ell$ in this order, such that $t_1$ is adjacent with $t$, and $t_\ell$
is adjacent with $t'$. Furthermore, to every node $t_i$ add a second child $s_i$ that is a leaf of the tree $T$.
The new node $s_i$ is of a new type, namely \textbf{edge leaf}: we set $\tdbag(s_i) = \tdbag(t')$, $V(G(s_i)) = \tdbag(t')$, and $E(G(s_i)) = \{e_i\}$.
The node $t_i$ is a \textbf{join} node with $G(t_i) = \tdG(t')-\{e_1,e_2,\ldots,e_{i-1}\}$. 

This completes the description of the tree decomposition refinement.
By this step, we have obtained the following properties: at every introduce or forget node $t$ with child $t'$ we have $G(t) = G(t')$,
while at every join node $t$ with children $t_1$ and $t_2$ we have $E(G(t)) = E(G(t_1)) \uplus E(G(t_2))$.
Furthermore, note that we have introduced two nodes for every edge of $G$, giving $\Oh(nk)$ new nodes in total.

We compute a labelling function $\Lambda: V(G) \to [k]$ such that $\Lambda$ is injective on every bag of $(T,\tdbag)$; such a labelling $\Lambda$
is straightforward to compute in a top-to-bottom fashion, using the fact that every bag of $(T,\tdbag)$ is of size at most $k$.

With the refined tree decomposition, we perform a bottom-up dynamic programming algorithm.
Fix a node $t \in V(T)$. For a subset $A \subseteq \tdunder(t)$ and an embedding $\Emb$ of $G(t)[A]$, we can naturally equip $\Emb$ with a structure
of a $k$-boundaried embedding by assigning labels $\Lambda|_{\tdbag(t)}$ (which is an injective function).
We will denote such a $k$-boundaried embedding at node $t$ as $\tdemb(t, \Emb)$

At $t$, we maintain the following DP table. 
For every set $X \subseteq \tdbag(t)$ and for every nice $k$-boundaried embedding $\tEmb$ of genus at most $g$ that uses only labels of $\Lambda(\tdbag(t) \setminus X)$,
we remember a value $\DP[t, X, \tEmb] \in \{0,1,\ldots\} \cup \{+\infty\}$.
Intuitively, we would like the value $\DP[t, X, \tEmb]$ to indicate the minimum size of a deletion set $Y \subseteq \tdunder(t) \setminus \tdbag(t)$
such that $G(t)-(X \cup Y)$ can be embedded equivalently with $\tEmb$.
However, for sake of the formal argument, we strip the above intuition into only its essential parts, and define the following two properties
that we maintain:
\begin{description}
\item[(a finite value yields small deletion set)]
If $\DP[t, X,\tEmb] \neq +\infty$, then
for every $k$ boundaried embedding $\tEmb_0$ that uses only labels of $\Lambda(\tdbag(t)\setminus X)$
and for every merge $\tEmb_{0M}$ of $\tEmb$ and $\tEmb_0$,
there exists a set $Y \subseteq \tdunder(t) \setminus \tdbag(t)$ of size at most $\DP[t, X, \tEmb]$,
an embedding $\Emb$ of $G(t)-(X \cup Y)$, 
and a merge $\tEmb_M$ of $\tdemb(t, \Emb)$ and $\tEmb_0$
of genus not larger than the genus of $\tEmb_{0M}$.
\item[(a good deletion set is represented in the table)]
for every set $Y \subseteq \tdunder(t) \setminus \tdbag(t)$,
every embedding $\Emb$ of $G(t)-(X \cup Y)$,
every $k$-boundaried embedding $\tEmb_0$ that uses only labels of $\Lambda(\tdbag(t)\setminus X)$,
and every merge $\tEmb_{0M}$ of $\tdemb(t, \Emb)$ and $\tEmb_0$,
there exists a nice $k$-boundaried embedding $\tEmb$
and a merge $\tEmb_M$ of $\tEmb$ and $\tEmb_0$
such that $\DP[t, X, \tEmb] \leq |Y|$
and the genus of $\tEmb_M$ is not larger than the genus of $\tEmb_{0M}$.
\end{description}
The above invariants include only one inequality from the definition of equivalence:
the one that is needed for the proof of the correctness of the algorithm.

In particular, 
since the root node $r$ has an empty bag, the minimum value of $\DP[r, \emptyset, \cdot]$ at the root node is the size of the minimum solution to \genusVD{} on $(G,g)$.
Indeed, 
by the second invariant, an optimum solution $Y^\ast$ and the corresponding embedding $\Emb^\ast$ has its corresponding entry $\tEmb^\ast$ when merged
with an empty embedding $\tEmb_0$ (so that $\tEmb_{0M} = \tdemb(r, \Emb) = (\Emb, k, \emptyset)$), and we have $|Y^\ast| \geq \DP[r, \emptyset, \tEmb^\ast]$.
In the other direction, the first invariant ensures that for every entry $\DP[r, \emptyset, \tEmb]$ and again an empty embedding $\tEmb_0$
and the only possible merge $\tEmb_{0M} = \tEmb$, there exists a deletion set $Y$ of size at most $\DP[r, \emptyset, \tEmb]$ and an embedding $\Emb$ of $G-Y$
such that $\tdemb(r, \Emb)$ has genus at most the genus of $\tEmb_{0M} = \tEmb$, which is at most $g$.

We now show how to compute the values $\DP[t, X,\tEmb]$ in a bottom-to-up fashion.

The computation is trivial at leaf nodes (empty bag, empty graph $G(t)$) and at edge leaf nodes (a single edge with distinct endpoints).
Furthermore, observe that at introduce nodes, we just copy the corresponding entries in $\DP$, as the new vertex is isolated in $G(t)$ and the
notion of an embedding ignores isolated vertices.
Thus, we are left with tackling forget and join nodes.

To this end, we define an operation of \emph{updating} a cell $\DP[t, X, \tEmb]$ with a value $a$: we set $\DP[t, X, \tEmb] := \max(a, \DP[t, X, \tEmb])$.
We say that the update was successful if the stored value changed.

\subsection{Forget nodes}
Let $t$ be a forget node with child $t'$ and let $v$ be the forgotten vertex. Note that, due to the refinement step,
we have $G(t') = G(t)$. The only difference between the nodes $t$ and $t'$ is that if we consider some embedding $\Emb$ of a subgraph of $G(t)$,
then $\tdemb(t', \Emb)$ assigns a label $\Lambda(v)$ to the vertex $v$, while in $\tdemb(t, \Emb)$ the vertex $v$ remains unlabelled.

This is straightforward to implement in our DP tables. We start by setting all values $\DP[t, X, \tEmb]$ to $+\infty$.
For every cell $\DP[t', X, \tEmb]$ with $v \in X$, we update the cell $\DP[t, X \setminus \{v\}, \tEmb]$ with value $\DP[t', X, \tEmb]+1$: when $v$ is deleted,
the embedding does not change, but we need to account for the vertex $v$ in the value of the cell.
For every cell $\DP[t', X, \tEmb']$ with $v \notin X$, we create an embedding $\tEmb$ from $\tEmb'$ by first forgetting the label of the vertex $v$ (so it becomes
    an unlabelled one), and then by applying Lemma~\ref{lem:make-nice} to obtain an equivalent embedding. We update the cell $\DP[t, X, \tEmb]$ with value
$\DP[t', X, \tEmb']$. 

We now formally verify that the entries for node $t$ satisfy the required properties.

\paragraph{First property.}
Take a cell $\DP[t, X, \tEmb]$ that is less than $+\infty$, a $k$-boundaried embedding $\tEmb_0$ and a merge $\tEmb_{0M}$ of $\tEmb$ and $\tEmb_0$.
Let $\DP[t', X', \tEmb']$ be the last cell that caused an update in the value of $\DP[t, X, \tEmb]$.

First, assume $v \in X'$; then $X' = X \cup \{v\}$ and $\tEmb' = \tEmb$. By the first property, applied inductively to the cell $\DP[t', X', \tEmb]$,
the embedding $\tEmb_0$ and the merge $\tEmb_{0M}$, we obtain a set $Y'$ of size at most $\DP[t', X', \tEmb]$, an embedding $\Emb'$
of $G(t')-(X' \cup Y')$ and a merge $\tEmb_M'$ of $\tdemb(t', \Emb')$ and $\tEmb_0$ of genus at most the genus of $\tEmb_{0M}$.
Set $Y = Y' \cup \{v\}$; then $X \cup Y = X' \cup Y'$ while $|Y| = |Y'| + 1 \leq \DP[t, X, \tEmb]$.
Consequently, $\Emb'$ is an embedding of $G(t)-(X \cup Y)$, and together with $Y$ and $\tEmb_M'$ satisfies the requirements of the first property.

In the second case, we have $v \notin X'$, and $X = X'$.
Let $\tEmb^\circ$ be the embedding $\tEmb'$ with the label of $v$ dropped; recall that $\tEmb$ is a result of application of Lemma~\ref{lem:make-nice}
to the embedding $\tEmb^\circ$. 
Since $\tEmb^\circ$ is equivalent to $\tEmb$, the genus-minimum merge $\tEmb_M^\circ$ of $\tEmb^\circ$ and $\tEmb_0$ has genus not larger than the genus
of $\tEmb_{0M}$.
Since neither $\tEmb^\circ$ nor $\tEmb_0$ uses the label $\Lambda(v)$, 
we can define $\tEmb_M'$ to be $\tEmb_M^\circ$ with additionally label $\Lambda(v)$ assigned to the vertex $v$;
note that this operation does not change the genus of the embedding.
As $\tEmb_0$ does not use the label $\Lambda(v)$, $\tEmb_M'$ is a merge of $\tEmb'$ and $\tEmb_0$. 
The first property, applied to the cell $\DP[t', X', \tEmb']$, $\tEmb_0$, and the merge $\tEmb_M'$,
yields existence of a set $Y' \subseteq \tdunder(t') \setminus \tdbag(t')$, an embedding $\Emb$ of $G(t')-(X'\cup Y')$, and a merge $\tEmb_M$
of $\tdemb(t', \Emb)$ and $\tEmb_0$ of genus at most the genus of $\tEmb_M'$.
Since $G(t') = G(t)$, $X' = X$, and $\tEmb_0$ does not use the label $\Lambda(v)$, $\tEmb_M$ is a merge of $\tdemb(t, \Emb)$ and $\tEmb_0$.
Furthermore, since the genus of $\tEmb_M'$ is at most the genus of $\tEmb_{0M}$, the tuple $(Y', \Emb, \tEmb_M)$ fulfills the first property.

\paragraph{Second property.}
Take a set $Y$, $\Emb$, $\tEmb_0$, and $\tEmb_{0M}$ as in the statement of the second property.
If $v \in Y$, take $Y' = Y \setminus \{v\}$ and $X' = X \cup \{v\}$, and otherwise take $Y' = Y$ and $X' = X$. In both cases, we have $X \cup Y = X' \cup Y'$,
$X' \subseteq \tdbag(t')$, and $Y' \subseteq \tdunder(t') \setminus \tdbag(t')$.

Since $\tEmb_0$ does not use the label $\Lambda(v)$, $\tEmb_{0M}$ is also a merge of $\tdemb(t', \Emb)$ and $\tEmb_0$.
By the second property for $t'$, $X'$, $Y'$, $\tEmb_0$, and $\tEmb_{0M}$, we obtain a $k$-boundaried embedding $\tEmb'$
and a merge $\tEmb_M'$ of $\tEmb'$ and $\tEmb_0$ of genus at most the genus of $\tEmb_{0M}$ such that $\DP[t', X', \tEmb'] \leq |Y'|$.
While processing the cell $\DP[t', X', \tEmb']$ the algorithm attempts an update on a cell $\DP[t, X, \tEmb]$ for $\tEmb$ being a result
of an application of Lemma~\ref{lem:make-nice} to the embedding $\tEmb^\circ$ being the embedding $\tEmb'$ with the label of $v$ dropped.
A direct check shows that in regardless of whether $v$ belongs to $Y$ or not, it follows that $\DP[t, X, \tEmb] \leq |Y|$.

Let $\tEmb_M^\circ$ be the embedding $\tEmb_M'$ with the label $\Lambda(v)$ dropped if present; note that the genera of $\tEmb_M^\circ$ and $\tEmb_M'$ are equal.
Since the embedding $\tEmb_0$ does not use the label $\Lambda(v)$, the merge $\tEmb_M^\circ$ is a merge of $\tEmb^\circ$ and $\tEmb_0$.
By the equivalence of $\tEmb^\circ$ and $\tEmb$, the genus-minimum merge $\tEmb_M$ of $\tEmb$ and $\tEmb_0$ is of genus at most the genus of $\tEmb_M^\circ$,
which in turn is at most the genus of $\tEmb_{0M}$. Thus, $\tEmb$ and $\tEmb_M$ fulfill the second property for $Y$, $\Emb$, $\tEmb_0$, and $\tEmb_{0M}$.

This finishes the description and the proof of correctness of the computations at forget nodes.

\subsection{Join nodes}
Let $t$ be a join node with children $t_1$ and $t_2$. Note that we have $V(G(t_1)) \cap V(G(t_2)) = \tdbag(t)$, $V(G(t_1)) \cup V(G(t_2)) = V(G(t))$
and $E(G(t)) = E(G(t_1)) \uplus E(G(t_2))$. 

We iterate over every set $X \subseteq \tdbag(t)$ and every pair of cells $\DP[t_1, X, \tEmb_1]$ and $\DP[t_2, X, \tEmb_2]$.
Note that both $\tEmb_1$ and $\tEmb_2$ are $k$-boundaried embeddings with labels on vertices of $\tdbag(t) \setminus X$. 
For every $k$-boundaried embedding $\tEmb'$ that is a merge of $\tEmb_1$ and $\tEmb_2$, we use Lemma~\ref{lem:make-nice} to compute a nice embedding
equivalent to $\tEmb'$, and update $\DP[t, X, \tEmb]$ with $\DP[t_1, X, \tEmb_1] + \DP[t_2, X, \tEmb_2]$.

We now formally verify that the entries for node $t$ satisfy the required properties.

\paragraph{First property.}
Take a cell $\DP[t, X, \tEmb]$ that is less than $+\infty$, a $k$-boundaried embedding $\tEmb_0$ and a merge $\tEmb_{0M}$ of $\tEmb$ and $\tEmb_0$.
Let $g_0$ be the genus of $\tEmb_{0M}$.
Assume the value of $\DP[t, X, \tEmb]$ comes from considering cells $\DP[t_i, X, \tEmb_i]$ and a merge $\tEmb'$ of $\tEmb_1$ and $\tEmb_2$.

By the equivalence of $\tEmb'$ and $\tEmb$, the genus-minimum merge $\tEmb_M^A$ of $\tEmb'$ and $\tEmb_0$ is of genus at most $g_0$.
Since $\tEmb'$ is a merge of $\tEmb_1$ and $\tEmb_2$, if we delete the edges of $\tEmb_2$ from the embedding $\tEmb_M^A$, we obtain
an embedding $\tEmb_1^A$ of not larger genus that is a merge of $\tEmb_1$ and $\tEmb_0$. Furthermore, observe that $\tEmb_M^A$ is a merge of $\tEmb_1^A$ and
$\tEmb_2$.

We apply the first property to the cell $\DP[t_2, X, \tEmb_2]$ with $\tEmb_1^A$ and $\tEmb_M^A$, obtaining a set $Y_2$, an embedding $\Emb_2$,
and a merge $\tEmb_M^B$ of $\tdemb(t_2, \Emb_2)$ and $\tEmb_1^A$ with genus at most $g_0$.
Recall that $\tEmb_1^A$ is a merge of $\tEmb_1$ and $\tEmb_0$. Consequently, if we delete the edges of $\tEmb_1$ from $\tEmb_M^B$ we obtain
an embedding $\tEmb_2^B$ that is a merge of $\tdemb(t_2, \Emb_2)$ and $\tEmb_0$. Furthermore, observe that $\tEmb_M^B$ is a merge of $\tEmb_2^B$
and $\tEmb_1$.

We now apply the first property to the cell $\DP[t_1,X, \tEmb_1]$ with $\tEmb_2^B$ and $\tEmb_M^B$, obtaining a set $Y_1$, an embedding $\Emb_1$,
and a merge $\tEmb_M$ of $\tdemb(t_1, \Emb_1)$ and $\tEmb_2^B$ of genus at most $g_0$.
Since $\tEmb_2^B$ is a merge of $\tdemb(t_2, \Emb_2)$ and $\tEmb_0$, if we delete the edges of $\tEmb_0$ from $\tEmb_M$,
we obtain an embedding $\tEmb_{12}$ that is a merge of $\tdemb(t_1, \Emb_1)$ and $\tdemb(t_2, \Emb_2)$. 
Let $\Emb$ be the embedding underlying $\tEmb_{12}$ and let $Y = Y_1 \cup Y_2$.
Then, $\Emb$ is an embedding of $G(t)-(X \cup Y)$ such that $\tEmb_M$ is a merge of $\tdemb(t, \Emb)$ with $\tEmb_0$.
Since the genus of $\tEmb_M$ is at most $g_0$, the proof of the first property is finished.

\paragraph{Second property.}
Take a set $Y$, $\Emb$, $\tEmb_0$, and $\tEmb_{0M}$ as in the statement of the second property, and let $g_0$ be the genus of $\tEmb_{0M}$.
For $i=1,2$, take $Y_i = Y \cap \tdunder(t_i)$ and $\Emb_i$ to be the embedding $\Emb$ restricted to the edges of $G(t_i)$
(i.e., with the edges of $G(t_{3-i})$ deleted).
Note that $\tdemb(t, \Emb)$ is a merge of $\tdemb(t_1, \Emb_1)$ and $\tdemb(t_2,\Emb_2)$.

First, let $\tEmb_1^A$ be the embedding $\tEmb_{0M}$ with the edges of $\Emb_2$ deleted. Observe that $\tEmb_1^A$ is a merge of $\tdemb(t_1, \Emb_1)$
and $\tEmb_0$, while $\tEmb_{0M}$ is a merge of $\tEmb_1^A$ and $\tdemb(t_2, \Emb_2)$.
By the second property, applied to $t_2$, $X$, $Y_2$, $\Emb_2$, $\tEmb_1^A$, and $\tEmb_{0M}$,
we obtain a nice $k$-boundaried embedding $\tEmb_2$ and a merge $\tEmb_M^B$ of $\tEmb_2$ and $\tEmb_1^A$ of genus at most $g_0$
such that $\DP[t_2, X, \tEmb_2] \leq |Y_2|$.

Now, if we delete the edges of $\Emb_1$ from $\tEmb_M^B$, we obtain an embedding $\tEmb_2^B$. Observe that
$\tEmb_2^B$ is a merge of $\tEmb_0$ and $\tEmb_2$, while $\tEmb_M^B$ is a merge of $\tdemb(t_1, \Emb_1)$ and $\tEmb_2^B$.
We apply the second property to $t_1$, $X$, $Y_1$, $\Emb_1$, $\tEmb_2^B$, and $\tEmb_M^B$,
obtaining a nice $k$-boundaried embedding $\tEmb_1$ and a merge $\tEmb_M$ of $\tEmb_1$ and $\tEmb_2^B$ of genus at most $g_0$
such that $\DP[t_1, X, \tEmb_1] \leq |Y_1|$.

Let $\tEmb$ be the embedding $\tEmb_M$ with the edges of $\tEmb_0$ deleted.
Since $\tEmb_M$ is a merge of $\tEmb_1$ and $\tEmb_2^B$, which in turn is a merge of $\tEmb_2$ and $\tEmb_0$,
we have that $\tEmb$ is a merge of $\tEmb_1$ and $\tEmb_2$ and $\tEmb_M$ is a merge of $\tEmb$ and $\tEmb_0$.
Consequently, the algorithm attempts an update on the cell $\DP[t, X, \tEmb]$ with the value
$$\DP[t_1, X, \tEmb_1] + \DP[t_2, X, \tEmb_2] \leq |Y_1| + |Y_2| \leq |Y|.$$
Since the genus of $\tEmb_M$ is at most $g_0$, the pair $\tEmb$ and $\tEmb_M$ fulfills the second property for $Y$, $\Emb$, $\tEmb_0$, and $\tEmb_{0M}$.

\subsection{Summary}
We have concluded description and correctness proofs of the operations at the nodes of the tree decomposition.
At every node $t$, a straightforward implementation takes time polynomial in the number of entries $\DP[t, X, \tEmb]$, which, by Corollary~\ref{cor:few-embeddings},
is $2^{\Oh((k+g)\log(k+g))}$. This finishes the proof of Theorem~\ref{thm:genusVDtw}.

\section{Irrelevant vertex}\label{sec:irr}

In this section, we prove Theorem~\ref{thm:irr}:
\irr*
The argumentation is heavily inspired by the corresponding planar case by Marx and Schlotter~\cite{ildi}. As in most irrelevant vertex arguments, we follow typical
outline:
\begin{enumerate}
\item Run an approximation algorithm for treewidth and, in case it returns
that the treewidth of the graph is larger than the required threshold, find a large
grid minor. Here, the linear dependency on the grid size and treewidth
in bounded genus graphs is known, and one can find the corresponding grid minor efficiently.
\item
Show that a vertex $w \in M$ that is connected to many places in the grid
that are far apart needs to be included in every solution, exactly
as it is in the planar case.
In the absence of such a vertex, a large part of the grid minor is flat, that is, embeds planarly
and does not have any internal connections to the modulator $M$.
\item Prove that a middle vertex of such a flat part is irrelevant. 
Here, the challenge is to argue that in every solution there exists an embedding of
the remaining part that draws the flat part indeed in a flat manner.
\end{enumerate}

We repeatedly use the following auxiliary result throughout this section:
\begin{lemma}[{\cite[Lemma B.6]{diestel}}]\label{lem:diestel}
Let $\Gamma$ be a surface of Euler genus $g$ and let $\mathcal{C}$ be a set of $g+1$ disjoint circles in~$\Gamma$. 
If $\Gamma \setminus \bigcup \mathcal{C}$ has a component $D_0$
whose closure in $\Gamma$ meets every circle in $\mathcal{C}$,
then at least one of the circles in $\mathcal{C}$ bounds a disc in $\Gamma$ that is disjoint from $D_0$.
\end{lemma}

Let $G$, $k$, $g$, and $M$ be such as in the statement of Theorem~\ref{thm:irr}.
We start by computing an embedding $\Emb_0$ of $G-M$ into a surface of Euler genus at most $g$, using either the
algorithm of Kawarabayashi, Mohar, and Reed~\cite{mohar2} or the older algorithm of Mohar~\cite{mohar1}.\footnote{The work~\cite{mohar2} appeared so far only as an extended abstract in conference proceedings.} This takes time $C_g^1 n$ for some constant $C_g^1$ depending only on $g$.

\subsection{Finding a large grid}

We apply a constant-factor approximation algorithm for treewidth and largest
excluded grid minor in graphs with a fixed excluded minor~\cite[Theorem~3.11]{bidim-algo}.
The algorithm does not need to compute a near-embedding of $G$, as $\Emb_0$ serves this purpose.
Thus, the algorithm as described in~\cite[Theorem~3.11]{bidim-algo} runs in time $C_g^2 n^{\Oh(1)}$ for some
constant $C_g^2$ depending only on $g$, and is a $C_g^3$-approximation.

If the resulting tree decomposition is of width at most 
$C_g^4 |M|^{1/2}k^{3/2}$ for some sufficiently large constant $C_g^4$, then we directly return it.
Otherwise,  by choosing the constant $C_g^4$ appropriately, we get a grid minor of sidelength $\Theta(|M|^{1/2}(k+g)^{3/2}g^{1/2})$.
By standard arguments, we henceforth focus on the case when we get a wall $W_0$ of sidelength $\ell_0$ as a subdivision in $G-M$, where $\ell_0= \Theta(|M|^{1/2}(k+g)^{3/2}g^{1/2})$ can be chosen arbitrarily large; see Figure~\ref{fig:wall} for the definition of a wall.

\begin{figure}
\begin{center}
\includegraphics{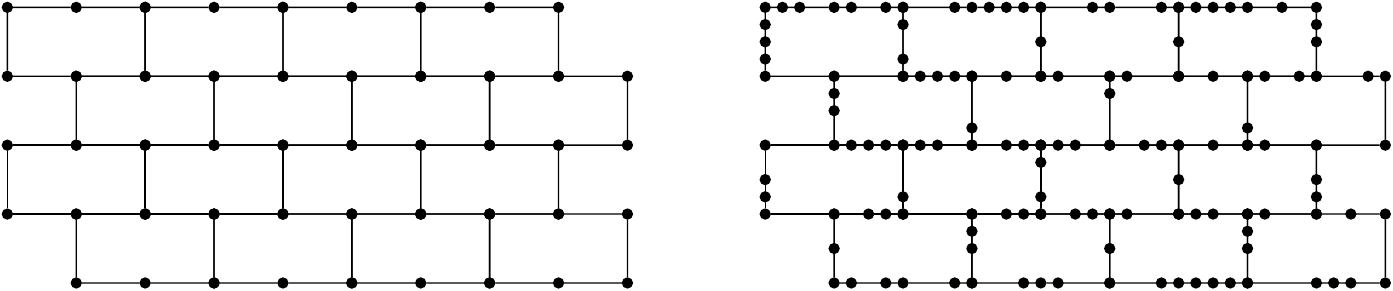}
\caption{A wall of sidelength $4$ and a subdivision of this wall.}
\label{fig:wall}
\end{center}
\end{figure}

In the wall $W_0$, we identify $g+1$ pairwise disjoint subwalls of sidelength
$\ell_1$ for every $\ell_1 = \Theta(\ell_0 /\sqrt{g}) = {\Theta(|M|^{1/2}(k+g)^{3/2})}$, 
leaving enough space between the subwalls so that the part of $W_0$ not contained in any of the subwalls is connected. 
By~\cref{lem:diestel}, for at least one of the 
identified subwalls, its surrounding cycle bounds a disc in the considered embedding of $G-M$. 
Consequently, at least one of the identified subwalls is planarly embedded in $\Emb_0$. 
We denote this subwall as $W_1$ and the part of $G-M$ embedded in the disc separated by the surrounding cycle of $W_1$ by $G_1$. 
In what follows, we mostly focus on $G_1$ and $W_1$.

\subsection{Filtering out the modulator neighbors}

We now show that if a vertex $w \in M$ is adjacent to many scattered
vertices of $G_1$, then it needs to be included in every solution.
We start with constructing a canonical graph of Euler genus larger than $g$.
\begin{lemma}\label{lem:K5-}
For an integer $\ell > 1$, let $B_\ell$ be a graph defined as follows.
We take $\ell$ disjoint copies of $K_5^-$ (i.e., the five-vertex clique $K_5$ with an edge deleted),
and denote by $v_i$ and $w_i$ the endpoints of the removed edge in the $i$-th copy.
We introduce two further vertices $v$ and $w$ and make $v$ adjacent to all vertices $v_i$ and $w$ adjacent to all vertices $w_i$.
Then, the Euler genus of $B_\ell$ is at least $\ell/2$.
\end{lemma}
\begin{proof}
Consider an embedding of $B_\ell$ into a surface $\Gamma$ of Euler genus $g$.
For $i=1,2,\ldots,\ell$, let $\gamma_i$ be the triangle in the $i$-copy of $K_5^-$ that does not pass through $v_i$ or $w_i$, treated as a closed curve on $\Gamma$.
Since $\Gamma$ has Euler genus $g$, by~\cref{lem:diestel}, at least $\ell-g$ of the curves $\gamma_i$ must be contractible.

Consider an index $i$ for which $\gamma_i$ is contractible, and let $V(\gamma_i)$ be the three vertices of $G$ that lie on $\gamma_i$.
Since $G-V(\gamma_i)$ is connected, $\gamma_i$ is a boundary of a face of the embedding of $G$.
Augment the graph $G$ by adding a new vertex $x_i$ inside this face, and connect it to the three vertices of $V(\gamma_i)$.

In this manner, the $i$-th copy of $K_5^-$, together with the new vertex $x_i$, forms a nonplanar graph (a supergraph of $K_{3,3}$). Consequently, we cannot perform this operation
in parallel for more than $g$ indices $i$. We infer that $\ell \leq 2g$, and the lemma is proven.
\end{proof}

Recall that the boundary cycle of $W_1$ encloses a disc containing $W_1$ in the considered embedding $\Emb_0$ of $G-M$, and $G_1$ is the part of the graph $G-M$ that is enclosed by this boundary cycle. 
For a vertex $v \in V(G_1)$ and an integer $\ell$, the \emph{radius-$\ell$ ball} $B(v, \ell)$ around $v$ is defined as follows: we take all faces $f$ of $W_1$
(excluding the infinite face) that contain $v$ either inside or on the boundary, mark all faces that lie in face-vertex distance (excluding traversal through the infinite face) at most $\ell$ 
from one of these faces $f$ and put into $B(v, \ell)$ all vertices of $G_1$ that are contained inside or on the boundary of marked faces. 

For every $w \in M$, we create a set $I(w) \subseteq V(G_1)$ as follows. We start with $I(w) = \emptyset$ and
every vertex of $V(G_1)$ unmarked. As long as there exists an unmarked
neighbor $v$ of $w$ in $G_1$, we insert $v$ into $I(w)$ and mark all vertices of $B(v, \ell_2)$ for $\ell_2 = 100(k+g+1)$.
We claim that if $I(w)$ is too large, the vertex $w$ needs to be deleted in any solution to \genusVD{}.
\begin{lemma}\label{lem:force-delete}
Let $w \in M$ be a vertex for which $|I(w)| > k+2g+1$. Then $w$ belongs to every solution to \genusVD{} instance $(G,g,k)$.
\end{lemma}
\begin{proof}
Suppose the contrary, and let $S$ be a solution that does not contain $w$. That is, $|S| \leq k$, and $G-S$ admits an embedding $\Emb$ into a surface
of Euler genus at most $g$.

Consider subgraphs $B(v, \ell_2/4)$ for $v \in I(w)$. By the construction of $I(w)$, these subgraphs are vertex-disjoint. 
Furthermore, as $\ell_2/4$ is much larger than $k$, for every $v_1,v_2 \in I(w)$, $v_1 \neq v_2$, there exists a set $\mathcal{P}(v_1,v_2)$ of $k+1$ vertex-disjoint paths
in $G_1$ connecting the outer boundary of $B(v_1, \ell_2/4)$ with the outer boundary of $B(v_2, \ell_2/4)$, without any internal vertex
in any of the subgraphs $B(v, \ell_2/4)$ for $v \in I(w)$. 
Let $I'$ be the set of these vertices $v \in I(w)$ such that $S$ is disjoint with $V(B(v, \ell_2/4))$; since $|S| \leq k$, we have $I' > 2g+1$.

Construct a minor $H$ of $G-S$ as follows. Pick some arbitrary $v_0 \in I'$ and contract $B(v_0, \ell_2/4)$ onto $v_0$.
For every $v \in I' \setminus \{v_0\}$, contract onto $v_0$ a path of $\mathcal{P}(v_0, v)$ that is disjoint with $S$.
For every $v \in I' \setminus \{v_0\}$, contract $B(v, \ell_2/4)$ into a $K_5^-$ (recall that $B(v, \ell_2/4)$ contains a large wall, being part of $W_1$)
such that one of the two nonadjacent vertices is adjacent to $w$ and the other to $v_0$. 
We have obtained that $B_{2g+1}$ (as defined in Lemma~\ref{lem:K5-}) is a minor of $G-S$, a contradiction with Lemma~\ref{lem:K5-}.
\end{proof}
Consequently, if there exists $w \in M$ with $|I(w)| > k +2g+1$, then we can return $w$ as the second result of the algorithm of Theorem~\ref{thm:irr}.
Henceforth we will assume that $|I(w)| \leq k+2g+1$ for every $w \in M$.

Recall that $W_1$ is a wall of sidelength $\ell_1$, where $\ell_1=\Theta(|M|^{1/2}(k+g)^{3/2})$ can be chosen arbitrarily large.
This lets us  identify $n_3$ disjoint subwalls of sidelength $\ell_3 = 100(k+g)$, where $n_3=\Theta(|M|(k+g))$ can be chosen arbitrarily large.
Note that $|\bigcup_{w \in M} I(w)| = \Oh(|M|(k+g))$ and that all neighbors of $M$ are located in $\bigcup_{w\in M, v\in I(w)}B(v,\ell_2)$.
Each ball $B(v,\ell_2)$ may intersect only a constant number of the identified subwalls of sidelength $\ell_3$.
Hence, if $n_3=\Theta(|M|(k+g))$ is sufficiently large,  there is a subwall $W_3$ of sidelength $\ell_3 = 100(k+g)$ such that no vertex of $G_3$,
defined as the part of $G_1$ that is enclosed by the outer boundary of $W_3$,  is a neighbor of a vertex in $M$. 
Note that only the vertices on the outer boundary of $G_3$ may have neighbors in $G-V(G_3)$.
With $G_3$ and $W_3$, we proceed to the next section.

\subsection{Middle vertex of a flat part is irrelevant}

Let $v$ be a vertex on the middle face of $W_3$. We show that $v$ is irrelevant, that is, it can be returned as the outcome of Theorem~\ref{thm:irr}.
Clearly, if $(G,g,k)$ is a yes-instance, then so is $(G-\{v\}, g, k)$. 
In the other direction, suppose that there exists $S \subseteq V(G) \setminus \{v\}$ such that $G-\{v\}-S$ admits an embedding $\Emb$ into a surface
of Euler genus at most $g$. We would like to enhance this embedding so that it also accommodates $v$.

The sidelength $\ell_3$ of $W_3$ is large enough to find in $W_3$ a subdivision of a \emph{circular wall} $W_4$ of height $h_4=5(k+g)+9$ and circumference $\ell_4=\min(3,k+1)$
so that $v$ is located inside the inner cycle of the wall; see \cref{fig:ring} for the definition of a circular wall.
Next, in $W_4$ we identify $k+g+2$ \emph{rings} (which are concentric circular walls of height 5 and circumference $\ell_4$), separated by layers of height 1; see \cref{fig:ring}.
Let $\mathcal{R}_0$ be the family of identified rings.

\begin{figure}
\begin{center}
\includegraphics{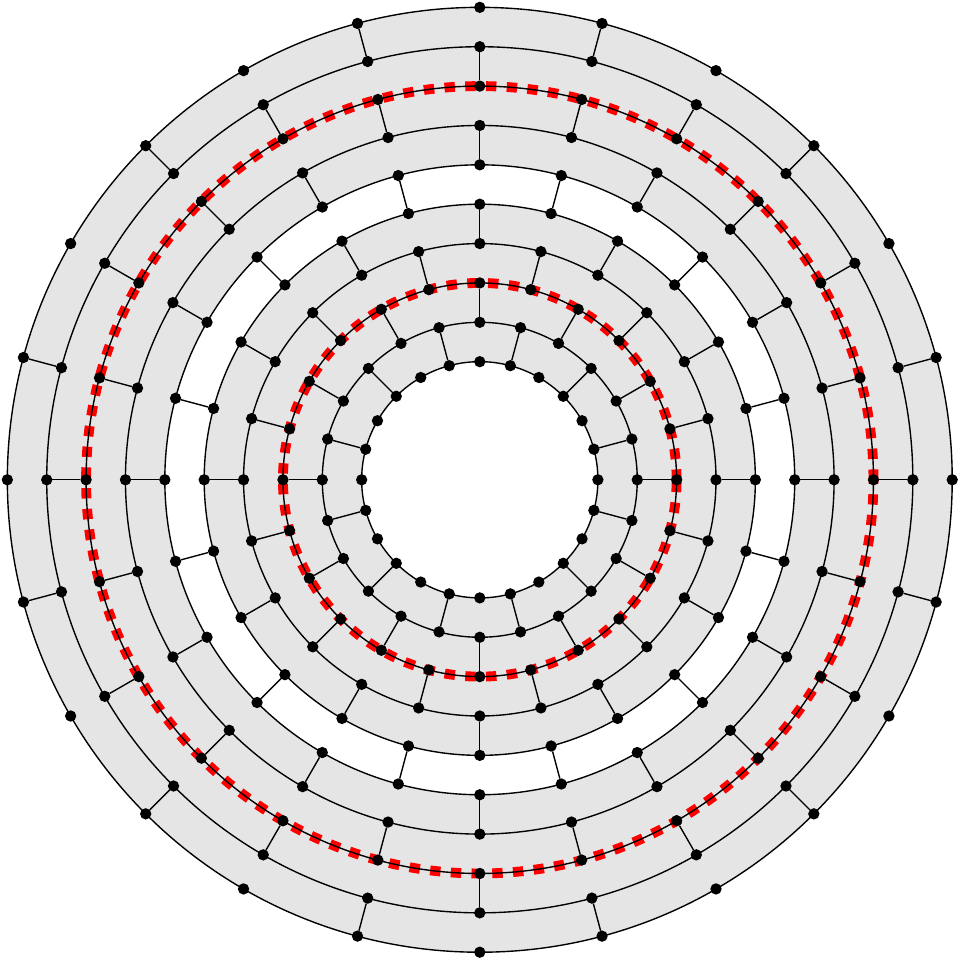}
\caption{A circular wall of height 9 and circumference 12 containing two rings. Their central circles are marked with red squares.}\label{fig:ring}
\end{center}
\end{figure}

For a ring $R\in \mathcal{R}_0$, we say that the \emph{central circle} of $R$ is the circle between the two middle layers of the ring,
and the \emph{boundary circles} are the circles separating $R$ from the remainder of $W_4$.
A face of $R$ in the embedding $\Emb_0$ is \emph{small} if it is not the outerface nor the face
inside the innermost cycle, and \emph{central} if it is incident to the central circle.
 The \emph{territory} of $R$, denoted by $T(R)$, is the subgraph of $G$ that consists of $R$ and 
everything that is drawn in the embedding $\Emb_0$ in the small faces of $R$.
Note that $T(R)$ is a planar graph, and the subgraphs $\{T(R) : R \in \mathcal{R}_0\}$ are vertex-disjoint.

Let $\mathcal{R} \subseteq \mathcal{R}_0$ be the family of these rings for which $T(R)$ is disjoint with $S$.
As $|\mathcal{R}_0| = k+g+2$ and the subgraphs $T(R)$ are vertex-disjoint, we have $|\mathcal{R}| \geq g+2$.

We say that a ring $R$ is embedded \emph{plainly} in the embedding $\Emb$ if
for every central face $f$ of $R$, the cycle $C^f$ that surrounds $f$ in $\Emb_0$, is a two-sided cycle in $\Emb$ that bounds
a disc on one side, and the graph $R-V(C^f)$ is drawn on the other side.
The following lemma is an easy corollary of~\cref{lem:diestel}:
\begin{lemma}\label{lem:plainly}
There exists $R \in \mathcal{R}$ that is embedded plainly in $\Emb$.
\end{lemma}
\begin{proof}
Suppose the contrary.
For every ring $R \in \mathcal{R}$, let $C(R)$ be the cycle around a small face of $R$ that violates the definition of a plainly embedded ring.
Consider a graph $G'=G-M-S-\bigcup_{R \in \mathcal{R}} V(C(R))$. We shall prove that it has a large connected component $D$ which contains $R':=R\setminus V(C(R))$ for each $R\in \mathcal{R}$.
First, note that $R'$ is itself connected and the boundary circles of $R$ of are preserved in $R'$. 
Thus, it suffices to prove that for every two consecutive rings $R_1,R_2\in \mathcal{R}$, the inner boundary circle of $R_2$
is connected to the outer boundary circle of $R_1$.
Observe that the part of $W_4$ between these circles is a circular wall of circumference $\ell_4$ and some positive height.
Thus, there are $\ell_4\ge k+1$ vertex-disjoint paths between the two boundary circles. These paths are disjoint with $M$ and $\bigcup_{R \in \mathcal{R}} V(C(R))$,
and at least one of them must be disjoint with $S$. Hence, one of these paths is preserved in $G'$.

We conclude that the claimed component $D$ indeed exists.
It is adjacent to every cycle $C(R)$, so \cref{lem:diestel}, due to our assumption on how cycles $C(R)$ are embedded in $\Emb$, yields that the Euler genus of the embedding $\Emb$ is at least $|\mathcal{R}|$, which is more than $g$.
\end{proof}

Let $R \in \mathcal{R}$ be a plainly embedded ring, and let $C_R$ be the central circle of $R$.
A direct corollary of the definition of a plainly embedded ring is the following.
\begin{corollary}\label{cor:CR}
In the embedding $\Emb$, $C_R$ is a two-sided cycle, and its incident edges of $R$ are partitioned between the sides of $C_R$ exactly as in the embedding
$\Emb_0$.
\end{corollary}
\begin{proof}
Let $e_1$ and $e_2$ be two edges of $R\setminus E(C_R)$ that are incident to $C_R$ that are two consecutive edges on the same side of $C_R$ in $\Emb_0$.
Furthermore, let $P$ be the path between $e_1$ and $e_2$ in $C_R$ that is not incident to any other edge of $R\setminus E(C_R)$ on the same side as $e_1$,
and let $f$ be the face of the embedding of $R$ in $\Emb_0$ that is incident with $e_1$, $e_2$, and $P$; see also Figure~\ref{fig:ringzoom}.
Let $e$ be the other edge of $R\setminus E(C_R)$ incident to $P$ (the one that is drawn on the opposite side of $C_R$ than $e_1$ in $\Emb_0$).
Since $R$ is plainly embedded in $\Emb$, the embedding $\Emb$ restricted to $R$ has a face surrounded by $C^f$, and thus, as we traverse $P$ in $\Emb$,
the edges $e_1$ and $e_2$ are on one side, and $e$ is on the other side. Since the choice of $e_1$, $e_2$, and $f$ is arbitrary, the claim follows.
\end{proof}

\begin{figure}
\begin{center}
\includegraphics{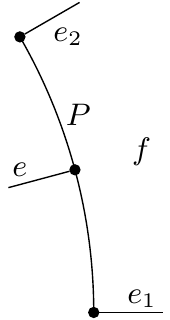}
\caption{Illustration for the proof of Corollary~\ref{cor:CR}.}
\label{fig:ringzoom}
\end{center}
\end{figure}

Corollary~\ref{cor:CR} allows us to speak about the \emph{inner} and \emph{outer} side of $C_R$ in $\Emb$: the sides of $C_R$ that contain incident edges inside $C_R$
and outside $C_R$ in $\Emb_0$.

A \emph{bridge} is a connected component $C$ of $G-V(R)$, together with the edges joining $C$ with $V(R)$. Furthermore, an edge $e \notin E(R)$
with both endpoints in $V(R)$ is also a bridge on its own. For a bridge $B$, the vertices of $V(B) \cap V(R)$ are \emph{attachment points}.

A bridge $B$ is \emph{central} if it has at least one attachment point, but all its attachment point lie in $V(C_R)$.
Note that, since $R$ is contained in $W_3$, a central bridge is disjoint with $M$ and in the embedding $\Emb_0$
it is drawn inside one of the small faces of $R$ incident with $C_R$. In particular, a central bridge is a subgraph of the territory of $R$.

Let $H$ be the subgraph of $G-M$ that consists of: the part of $G-M$ that is enclosed on the same side of $C_R$ in the embedding $\Emb_0$ as the vertex $v$
(i.e., on the disc, flat side of $C_R$), together with all central bridges.
In other words, $H$ is a subgraph of $G-M$ induced by the vertices of $C_R$ and all connected components of $G-M-V(C_R)$ that are contained
in the same connected component of $G-M$ as $C_R$, except for the connected component of $G-M-V(C_R)$ that contains the outermost concentric cycle of $R$.
The boundary of $H$, denoted $\partial H$, is the set of vertices of $H$ that have incident edges of $G$ that do not belong to $H$.

Armed with these observations, 
we now modify the embedding $\Emb$ of $G-S-\{v\}$ as follows. First, we cut the surface along the cycle $C_R$, and cap with discs the two resulting holes.
This operation can only lead to a surface of a lower Euler genus.
Second, we remove all edges of $H$ from $\Emb$, and all vertices of $H$ that become isolated by this operation.
Observe that due to Corollary~\ref{cor:CR}, now in $\Emb$ there is a face $f$, containing one of the discs glued to $C_R$ (the one glued on the inner side of $C_R$).

Note that if we restrict the embedding $\Emb_0$ to $H$, we obtain a plane embedding $\Emb_H$ of $H$
with an additional property that every vertex of $\partial H$ lies on the outerface.

Consider a bridge $B$ that has an attachment point in $V(C_R)$, but is not central (i.e., is not a subgraph of $H$).
Note that due to the fact that $R$ is plainly embedded, in $\Emb$ the bridge $B$ needs to be drawn in the same face of $R$ as in the embedding $\Emb_0$.
Thus, the order of the vertices of $\partial H$ on the outerface of $\Emb_H$ is exactly the same as the order of these vertices in the modified embedding $\Emb'$.
Consequently, the embedding $\Emb_H$ of $H$ can be glued into $f$,
identifying correspondingly the vertices of $V(H)$ that remain in $\Emb$.
In this manner, we obtain an embedding $\Emb'$ of $G-(S \setminus V(H))$, concluding the proof that $v$ is irrelevant.

This finishes the proof of Theorem~\ref{thm:irr}.

\section{Conclusions}

In this work we have developed fixed-parameter algorithms for the \genusVD{} problem with solution size and treewidth parameterizations, putting particular effort
into optimizing the dependency on the treewidth parameter in the running time bound. 

We remark that, although our formal statement of \genusVD{} involves only bounding the Euler genus of the output graph, only minor changes to our algorithms are required
if one demands the final graph to be embeddable in an \emph{orientable} surface of some genus. 
In terms of combinatorial embeddings, studied in Section~\ref{sec:embed}, an embedding is orientable if the set of flags can be partitioned into two parts (called \emph{left} and \emph{right})
such that every orbit of $\Vperm$, $\Eperm$, and $\Fperm$ contains two flags from different sets. 
The crucial observation is that deleting an edge, drawing an edge along a face boundary, and suppressing a size-$4$ vertex that is not isolated, applied to an orientable embedding
results in an embedding that is also orientable. Consequently, if we allow only orientable embeddings in the dynamic programming algorithm of Section~\ref{sec:tw}, we obtain the desired variant
of Theorem~\ref{thm:genusVDtw} for orientable surfaces. Finally, the arguments of Section~\ref{sec:irr} operate in the language of modifying an embedding in a fixed surface, and therefore
yield also without any changes a variant of Theorem~\ref{thm:irr} for orientable surfaces.

We would like to conclude with two open questions, stemming from our research.

First: Can we obtain a $2^{\Oh(C_g k \log k)} n$-time algorithm, following the ideas of~\cite{jls} for the planar case?
Our bounded treewidth routine suits such an algorithm, but the irrelevant vertex argument does not.

Second, and more challenging: what can we say about possible dependency on the parameter $k$
for the problem of deleting $k$ vertices to an arbitrary minor-closed graph family?
A similar question can be asked for the parameter treewidth.
Here, the main challenge is that it is harder to certify being $H$-minor-free for an arbitrary
graph $H$, while one can certify being of bounded genus by giving a corresponding embedding.

\bibliographystyle{plainurl}
\bibliography{genus}

%\newpage
%\input{appendix.tex}

\end{document}